\documentclass[9pt]{sigplanconf}

\usepackage{balance}
\usepackage{amsmath}

\usepackage{amsthm}
\usepackage{citesort}

\usepackage{color}
\usepackage{url}

\usepackage{graphics,graphicx}
\usepackage{pstricks,pst-node,pst-tree}

\newcommand{\OMIT}[1]{}

\newcommand{\Lang}{\mathcal{L}}
\newcommand{\ialphabet}{\Sigma}

\newcommand{\Var}{\mathcal{V}}

\newcommand{\N}{\mathbb{N}}
\newcommand{\Z}{\mathbb{Z}}

\newcommand{\defn}[1]{\emph{#1}}

\OMIT{

}

\newtheorem{theorem}{Theorem}
\newtheorem*{remark}{Remark}
\newtheorem{example}{Example}

\newcommand{\PSPACE}{\ensuremath{\mathsf{PSPACE}}}
\renewcommand{\P}{\ensuremath{\mathsf{P}}}
\newcommand{\EXPSPACE}{\ensuremath{\mathsf{EXPSPACE}}}

\newcommand\anthony[1]{{\color{blue}
[#1 - \textbf{Anthony}]}}

\newcommand{\replace}[2]{\ensuremath{\text{\ttfamily replace-all}(#1,#2)}}
\newcommand{\String}{\ensuremath{\text{\textbf{str}}}}
\newcommand{\Rel}{\ensuremath{\text{\textbf{rel}}}}
\newcommand{\Int}{\ensuremath{\text{\textbf{int}}}}
\newcommand{\Char}{\ensuremath{\text{\textbf{char}}}}
\newcommand{\Reg}{\ensuremath{\text{\textbf{reg}}}}

\newcommand{\NewStrCon}[1]{\ensuremath{\lambda_{\Rel}(#1)}}
\newcommand{\NewCon}[1]{\ensuremath{\lambda(#1)}}
\newcommand{\NewRegCon}[1]{\ensuremath{\lambda_{\Reg}(#1)}}
\newcommand{\Aut}{\ensuremath{\mathcal{A}}}
\newcommand{\AutB}{\ensuremath{\mathcal{B}}}
\newcommand{\Transducer}{\ensuremath{T}}
\newcommand{\controls}{\ensuremath{Q}}
\newcommand{\finals}{\ensuremath{F}}
\newcommand{\transrel}{\ensuremath{\delta}}

\newcommand{\tran}[1]{\stackrel{#1}{\longrightarrow}}

\newcommand{\mcl}[1]{\mathcal{#1}}
\newcommand{\IndexOf}{\ensuremath{\text{\sf IndexOf}}}

\newcommand{\rarw}{\rightarrow}
\newcommand{\GraphR}{\mathbb{G}}

\newcommand{\Dim}{\text{Dim}}
\newcommand{\AC}{\ensuremath{\sf AC}}
\newcommand{\pspace}{\PSPACE}
\newcommand{\M}{{\cal M}} 
\newcommand{\Cons}{{\rm Cons}} 

\newcommand{\AUWR}{{\sf SL}} 
\newcommand{\AUWRe}{\AUWR^e} 
\newcommand{\AUR}{{\sf SL}_r} 
\newcommand{\AURe}{\ensuremath{\AUR^e}}

\newcommand{\sat}{{\sc Satisfiability}} 

\newcommand{\sigmas}{\Sigma^*}
\newcommand{\sigmabot}{\Sigma_\bot}

\newcommand{\sigbot}{\sigmabot}
\newcommand{\e}{\varepsilon}

\renewcommand{\l}{\ell}

\newcommand{\A}{{\cal A}} 
 
\newcommand{\G}{{\cal G}}

\newcommand\shortlong[2]{#2}
\usepackage{amsfonts}

\newtheorem{proposition}[theorem]{Proposition}
\newtheorem{lemma}[theorem]{Lemma}
\newtheorem{definition}{Definition}
\newtheorem*{convention}{Convention}
\newtheorem{future}{Future Work}

\newif\ifdraft\draftfalse
\ifdraft
\newcommand\anthonychanged[1]{{\color{blue}{#1}}}

\else
\newcommand\anthonychanged[1]{#1}

\fi

\begin{document}

\toappear{}

\setlength{\pdfpageheight}{\paperheight}
\setlength{\pdfpagewidth}{\paperwidth}

\conferenceinfo{To appear in POPL'16}{January 20--22, 2016, St. Petersburg,
FL, USA} 
\copyrightyear{2016} 
\copyrightdata{978-1-nnnn-nnnn-n/yy/mm} 
\doi{nnnnnnn.nnnnnnn}




\titlebanner{}        
\preprintfooter{String Solving with Word Equations and Transducers}   

\title{String Solving with Word Equations and Transducers: Towards a Logic
for Analysing Mutation XSS\shortlong{}{ (Full Version)}}

\authorinfo{Anthony W. Lin}
           {Yale-NUS College, Singapore}
           {anthony.w.lin@yale-nus.edu.sg}
\authorinfo{Pablo Barcel\'o}
           {Center for Semantic Web Research \\ \& Dept. of Comp. Science, Univ.
           of Chile, Chile}
           {pbarcelo@dcc.uchile.cl}

\maketitle

\begin{abstract}
\OMIT{
In the past seven years there has been a lot of work in string solving
especially with applications to detecting security vulnerabilities in
web applications against code 
injections and cross-site scripting (XSS). 
}
We study the fundamental issue of decidability of satisfiability over string
logics with concatenations and finite-state transducers as atomic operations.
Although restricting to one type of operations yields decidability,
%
little is known about the decidability of their combined theory, which is 
especially relevant when 
analysing security vulnerabilities of dynamic web pages in a more realistic 
browser model.
%
On the one hand, word equations (string logic with concatenations) cannot 
precisely capture sanitisation functions (e.g. htmlescape) and 
implicit browser transductions (e.g. innerHTML mutations).
On the other hand, transducers suffer from the reverse problem of being able to 
model sanitisation functions and browser transductions, but not string
concatenations.
%
Naively combining word equations and transducers easily leads to an undecidable 
logic. Our main contribution is to show that the 
``straight-line fragment'' of the logic is decidable (complexity ranges from 
PSPACE to 
EXPSPACE). The fragment can express the program logics of straight-line 
string-manipulating programs with concatenations and 
transductions as atomic operations, which arise when
performing bounded model checking or dynamic symbolic executions.
\OMIT{
can express formulas obtained after eliminating 
loops in 
programs (with string concatenations and transductions as atomic operations), 
and converting them into static single
assignment form. 
}
We demonstrate that the logic can naturally express constraints required
for analysing mutation XSS in web applications.
%
Finally, the logic remains decidable in the presence of length, 
letter-counting, regular, indexOf, and disequality constraints.
 
\end{abstract}






\section{Introduction}
The past decade has witnessed a significant amount of progress in constraint 
solving technologies, thanks to the emergence of highly efficient SAT-solvers
(e.g. see \cite{SAT-CACM,SAT-handbook,KS08})
and SMT-solvers (e.g. see \cite{SMT-CACM,SMT-chapter,KS08}). The goal of SAT-solvers
is to solve constraint satisfaction 
problem in its most basic form, i.e., satisfiability of propositional formulas. 
Nowadays there are numerous highly efficient solvers including Chaff,
Glucose, Lingeling, and MiniSAT, to name a few (e.g. see 
\cite{SAT-competition} for others). 
Certain applications, however, require more expressive constraint
languages. The framework of satisfiability modulo theories (SMT) builds on top
of the basic constraint satisfaction problem by interpreting
atomic propositions as a quantifier-free formula in a certain ``background''
logical
theory, e.g., linear arithmetic. Today fast SMT-solvers supporting a 
multitude of background theories are available including
Boolector, CVC4, Yices, and Z3, to name a few (e.g. see \cite{SMT-competition} 
for others).
The backbones of fast SMT-solvers are usually
a highly efficient SAT-solver (to handle disjunctions) and an effective
method of dealing with satisfiability of formulas in the background theory.

In the past seven years or so there have been a lot of works on developing
robust SMT-solvers for constraint languages over 
strings (a.k.a. \emph{string solvers}). The following (incomplete) list 
of publications indicates the amount of interests in string solving (broadly 
construed):
\cite{Stranger,fmsd14,Berkeley-JavaScript,S3,BTV09,Abdulla14,cvc4,HW12,RVG12,Yu09,Balzarotti08,patching,symbolic-transducer,ganesh-word,SMV12,Was07,Was08a,Was08b,FL10,FPBL13,BEK,HAMPI,Z3-str,yu2011,CMS03,GSD04,Min05,DV13}. 
One main driving force behind this
research direction --- as argued by the authors of
\cite{Berkeley-JavaScript,S3,Z3-str,symbolic-transducer,Balzarotti08,fmsd14,DV13,Min05,BEK,Was07,Was08a,Was08b,FPBL13,HAMPI,Min05,HAMPI} 
among others --- is the application to analysis of security vulnerabilities in 
web applications against code injections and cross-site scripting (XSS), which
are typically caused by improper handling of untrusted strings
by the web applications, e.g., leading to an execution of malicious JavaScript 
in the clients' browsers. 

Despite the amount of recent progress in developing practical string solvers, 
little progress has been made on the foundational issues of string solving
including \emph{decidability}, which is particularly important since it 
imposes a fundamental limit of what we can expect from a solver for a given 
string constraint language (especially with respect to soundness and 
completeness). Perhaps the most important theoretical result in string solving 
is the decidability of satisfiability
for \emph{word equations} (i.e. string logic with the concatenation operator) 
by Makanin \cite{Makanin} (whose computational complexity was improved 
to $\PSPACE$ by Plandowski \cite{Plandowski,Plandowski06}). It is known
that adding regular constraints (i.e. regular-expression matching) preserves 
decidability without increasing the complexity
\cite{Plandowski06,diekert}. Very few decidability
results extending this logic are known.

For such an application as detecting security vulnerabilities in web 
applications in a realistic browser model (e.g. see \cite{web-model}),
word equations with regular constraints alone are \emph{insufficient}. Firstly,
browsers regularly perform \emph{implicit transductions}. For example,
upon \texttt{innerHTML} assignments or \texttt{document.write} invocations, 
browsers mutate the original string values by HTML entity decoding, e.g.,
each occurrence of \texttt{\&\#34;} will be replaced by \verb+"+.
[Modern browsers admit some exceptions including the HTML entity names 
\texttt{\&amp;}, \texttt{\&lt;}, and \texttt{\&gt;} (among others), which
will not be decoded.] Since such transductions involve only conversions of
one character-set encoding to another, they can be encoded as \emph{finite-state
(input/output) transducers}, as has already been noted in 
\cite{web-model,BEK,symbolic-transducer,DV13} (among others).
Secondly, in an attempt to prevent code injection and XSS attacks, most web 
applications will first sanitise untrusted strings (e.g. obtained from an 
untrusted application) before processing them. Common sanitisation functions 
include JavaScript-Escape and HTML-Escape (an implementation can be found in
The Closure Library \cite{closure}). HTML-Escape converts reserved characters 
in HTML such as {\ttfamily \verb+&+}, {\ttfamily \verb+<+}, and {\ttfamily 
\verb+'+} to their respective HTML entity 
names {\ttfamily \verb+&amp;+}, {\ttfamily \verb+&lt;+} and 
{\ttfamily \verb+&#39;+}.
On the other hand, JavaScript-Escape will backslash-escape certain 
metacharacters,
e.g., the character {\ttfamily \verb+'+} and {\ttfamily \verb+"+} are replaced 
by {\ttfamily \verb+\'+} and {\ttfamily \verb+\"+}. Again, such sanitisation
functions can be encoded as finite-state transducers, as has been noted in
\cite{BEK} (among others). 
\begin{example}
    \label{ex:cacm}
    {\em 
The following JavaScript code snippet adapted from a recent CACM article 
\cite{Kern14} is a simple example that uses \emph{both} concatenations and 
finite-state transducers (both explicitly and implictly):
\begin{flushleft}
    \footnotesize
\begin{verbatim}
var x = goog.string.htmlEscape(cat);
var y = goog.string.escapeString(x);
catElem.innerHTML = '<button onclick=
    "createCatList(\'' + y + '\')">' + x + '</button>';
\end{verbatim}
\end{flushleft}
The code assigns an HTML markup for a hyperlink to the DOM element catElem. The 
hyperlink creates a category \texttt{cat} whose value is provided by an
untrusted third party. For this reason, the code attempts to first sanitise 
the value of \texttt{cat}.
This is done via The Closure Library \cite{closure} string functions 
\texttt{htmlEscape} and \texttt{escapeString} (implementing JavaScript-Escape). 
Inputting the value \texttt{Flora \& Fauna} into \texttt{cat} gives the desired
HTML markup:
\begin{flushleft}
    \footnotesize
        \begin{verbatim}
<button onclick="createCatList('Flora &amp; Fauna')">
    Flora &amp; Fauna</button>
    \end{verbatim}
\end{flushleft}
On the other hand, inputting the value {\ttfamily \verb+');alert(1);//+}
to \texttt{cat}, results in the HTML markup:
\begin{flushleft}
    \footnotesize
    \begin{verbatim}
<button onclick="createCatList('&#39;);alert(1);//')">
    &#39;);alert(1);//')</button>
    \end{verbatim}
\end{flushleft}
When this is inserted into the DOM via \texttt{innerHTML}, an implicit browser 
transduction will take place, i.e., first HTML-unescaping the value inside the 
\texttt{onclick} attribute and invoking the attacker's script 
\texttt{alert(1)} after \texttt{createCatList}.
This subtle XSS bug (a type of \emph{mutation XSS} \cite{mXSS}) is due 
to calling the appropriate escaping functions in the wrong order.
}
\qed
\end{example}

It is well-known (e.g. see \cite{Balzarotti08,fmsd14,Berkeley-JavaScript})
that string solving can be applied to detecting security 
vulnerabilities against a given injection and XSS attack pattern\footnote{Attack
patterns are identified from previous vulnerabilities, some of which
have been well-documented, e.g., see \cite{OWASP-XSS,html5sec}.}
$P$ in the form of a regular expression.
After identifying certain 
``hot spot'' variables in the program where attacks can be performed (e.g. 
possibly via taint analysis), a string constraint will be generated that is 
satisfiable
iff the program is vulnerable against a given attack pattern. In the above
example, to analyse security vulnerabilities in the variable
\verb+catElem.innerHTML+ against the following attack pattern (given 
in JavaScript regex notation; blank space inserted for readability):
{\footnotesize 
\begin{verbatim}
e1 = /<button onclick=
        "createCatList\(' ( ' | [^']*[^'\\] ' ) \); 
            [^']*[^'\\]' )">.*<\/button>/
\end{verbatim}}
\noindent
one would express the program logic as a conjunction of:
\begin{itemize}
    \item $x = R_1(\text{\texttt{cat}})$
    \item $y = R_2(x)$
    \item $z = w_1 \cdot y \cdot w_2 \cdot x \cdot
        w_3$ for some constant strings $w_1, w_2, w_3$, e.g., 
        $w_1$ is \verb+<button onclick="createCatList('+
    \item $\text{\texttt{catElem.innerHTML}} = R_3(z)$
    \item $\text{\texttt{catElem.innerHTML}}$ matches \text{\texttt{e1}}.
\end{itemize}
Here, $R_1$ and $R_2$ are, respectively, transducers implementing 
\texttt{htmlEscape} and \texttt{escapeString}, while $R_3$ is a transducer
implementing the implicit browser transductions upon \texttt{innerHTML}
assignments. Note that the above string constraint cannot be written as
word equations with regular constraints alone since the finite-state transducers
replace \emph{each occurrence} of a substring (e.g. \verb+&#39;+) in a string 
by another string (e.g. the single character \verb+"+). To the best of our
knowledge, there is no known decidable logic which can express the above
string constraint.

\paragraph{{\bf Contribution:}} We study the decidability of satisfiability over 
string logics with concatenations, finite-state transductions, and regex
matching as atomic operations. Naively combining concatenations and 
transducers easily leads to an undecidable logic (e.g. see \cite{BTV09,BFL13}). 
In fact, it was 
shown in \cite{BFL13} that restricting to string constraints of the form
$x = y\cdot z \wedge x = R(z)$, where $R$ ranges over finite-state transducers,
is undecidable. [Actually, $R$ can be restricted to a relatively weak class of 
finite-state transducers that express only \emph{regular relations} (a.k.a. 
\emph{automatic relations} \cite{BG04}).] 
Our main contribution is to show that the ``straight-line
fragment'' of the logic is decidable (in fact, is $\EXPSPACE$-complete, but 
under a certain reasonable assumption the complexity reduces to $\PSPACE$). 
In fact, our decidability results provide an upper bound for the maximum size of
solutions that need to be explored to guarantee completeness for bounded-length
string solvers whenever the input constraint falls within our straight-line
fragment. The 
fragment can express the program 
logics of straight-line string-manipulating programs with concatenations and 
transductions as atomic operations. This includes
the program logic of Example \ref{ex:cacm}. 
In fact, straight-line programs naturally arise when performing \emph{bounded 
model 
checking}\footnote{Note that this does \emph{not} mean restrictions to strings 
of bounded length.} or \emph{dynamic symbolic executions}, which
unrolls loops in the programs (up to a given depth) and converts the resulting 
programs into \emph{static single assignment form} (i.e. each variable defined 
once). Please consult \cite{S3,KS08,DKW08} for more details.

Example \ref{ex:cacm} is one example of analysis of mutation XSS vulnerabilities
that can be expressed in our logic. In this paper, we provide three other
mutation XSS examples 
that can be expressed in our logic (adapted from \cite{Kern14,Stock14,mXSS}).

Finally, the case study of
\cite{Berkeley-JavaScript} suggested
that the use of length constraints (comparing lengths of different strings) and 
$\IndexOf$
constraints ($\IndexOf(x,y)$ outputs the position of an occurrence of a given 
string $x$ in another string $y$) is prevalent in JavaScript 
programs. 
To this end, we show that our logic is still decidable (with the
same complexity) when extended with:
\begin{enumerate}
    \item \emph{(linear) arithmetic constraints} 
        whose free variables are interpreted
        as one of the following: integer variables in the program, length of a 
        string variable, or the number of occurrences of a certain letter
        in a string variable.
    \item \emph{$\IndexOf$ constraints} of the form $h = \IndexOf(w,y)$, where
        $w$ is a constant string and $h$ is an integer variable. This is the most 
        frequent usage of $\IndexOf$ operator in JavaScript.
\end{enumerate}
All the examples in the benchmark examples of \cite{Berkeley-JavaScript} were 
observed
to belong to a class called \emph{solved forms} \cite{ganesh-word}, which is a 
subset of our string logic with linear arithmetic constraints. Lastly,
we can also add unrestricted disequality constraints between string variables,
while still preserving decidability.

\paragraph{{\bf Organisation: }}
Section \ref{sec:prelim} fixes general mathematical notations and reviews 
the necessary concepts on automata and transducers that will be used 
throughout the paper. In Section \ref{sec:lang}, we
define a general string constraint language that combines
concatenations, finite-state transductions, and regex matching. 
The language is undecidable even after imposing various restrictions 
that have been proposed in the literature.
In Section 
\ref{sec:decidable_string}, we recover decidability for the straight-line
fragment. 
In Section \ref{sec:ext}, we show that decidability can be retained even when 
length and $\IndexOf$
constraints are incorporated. We conclude with the related work and possible
future works in Section \ref{sec:rw}. Additional material can be found in
the \shortlong{full version}{appendix}.

\section{Preliminaries}
\label{sec:prelim}

\OMIT{
Constraint languages over strings often have roots in automata and formal
language theory. For this reason, this section aims to
review the definitions and
basic results on finite-state automata and input/output transducers, 
which our constraint language will be based on. 
}

\paragraph*{{\bf General notations:}}
Given two integers $i, j$, we write $[i,j]$ to denote the set
$\{i,\ldots,j\}$ of integers in between $i$ and $j$. For each integer $k$, we 
write $[k]$ to denote $[0,k]$.
Given a binary relation $R \subseteq S \times S$ and a set $A \subseteq S$, we
use $R(A)$ to denote the set $\{s' \in S \mid (s,s') \in R \text{ for some
$s \in A$}\}$. 
In other words, $R(A)$ is the post-image of $A$ under $R$. 
We use $R^{-1}$ to denote the reverse relation, i.e., 
$(s,s') \in R^{-1}$ iff $(s',s) \in R$ for each $s,s' \in S$. 
Notice then that $R^{-1}(A)$ is the pre-image of $A$ under $R$.  
The term \defn{DAG} stands for \emph{directed acylic graphs}.



\paragraph*{{\bf Regular languages:}}
Fix a finite alphabet $\Sigma$. Elements in $\Sigma^*$ are interchangeably called words or strings. 
For each finite word  
$w = w_1\ldots w_n \in 
\Sigma^*$, we
write $w[i,j]$, where $1 \leq i \leq j \leq n$, to denote the segment
$w_i\ldots w_j$. We use $w[i]$ to denote $w[i,i] = w_i$. In addition,
the symbol $|w|$ denotes the length $n$ of $w$, while the symbol $|w|_a$ 
$(a \in \ialphabet)$ denotes the number of occurrences of $a$ in $w$.

Recall that a {\em nondeterministic
finite-state automaton} (NFA) is a tuple 
$\mcl{A} = (\Sigma,Q,\delta,q_0,F)$, where $Q$ is a finite set of states, $\delta 
\subseteq Q \times \Sigma \times Q$ is the transition relation, $q_0
\in Q$ is the initial state, and $F \subseteq Q$ is the set of final
states. 
A {\em run} of $\mcl{A}$ on $w$ is a function $\rho: \{0,\ldots,n\}
\rarw Q$ with $\rho(0) = q_0$ that obeys the transition relation $\delta$,
i.e., $(\rho(i),w_i,\rho(i+1)) \in \delta$ for each $i \in \{0,\ldots,n-1\}$. 
We write $\Aut_{[q,q']}$ to denote $\Aut$ but the initial state (resp. set of 
final states) is
replaced by $q$ (resp. $\{q'\}$).
We may also denote the run $\rho$ by the word $\rho(0)\cdots \rho(n)$ over
the alphabet $Q$. 
The run $\rho$ is said to be \defn{accepting} if $\rho(n) \in F$, in which
case we say that the word $w$ is \defn{accepted} by $\mcl{A}$. The language
$L(\mcl{A})$ of $\mcl{A}$ is the set of words in $\Sigma^*$ accepted by
$\mcl{A}$. Such a language is said to be \defn{regular}. 
Recall that regular languages are precisely the ones that can be defined by 
regular expressions. 
In the sequel, when the meaning is clear from the context, we will sometimes 
confuse an NFA or a regular expression with the regular language that it
recognises/generates.


\paragraph*{{\bf Transducers and rational relations:}} 

A  transducer (short for ``finite-state input output transducer'') is a 
two-tape automaton that has two 
heads for the tapes and one 
additional finite control; at every step, based on the state and the letters it 
is reading, the automaton can enter a new state and move some (but not
necessarily all) tape heads. Each transducer generates a binary relation over
strings called rational relation. 

We will now make this definition more precise. A \defn{transducer} over the 
alphabet $\ialphabet$ is a tuple $\Aut = 
(\Gamma,\controls,\transrel,q_0,\finals)$, where
$\Gamma := \ialphabet_\epsilon^2$ and $\ialphabet_\epsilon := \ialphabet
\cup \{\epsilon\}$, such that $\Aut$ is syntactically an NFA over $\Gamma$.
The transducer $\Aut$ is said to be \defn{synchronised} if $\Aut$ (viewed as
an NFA) does not accept words $w = (a_1,b_1)\cdots (a_n,b_n) \in
(\ialphabet_\epsilon^2)^*$ such that there exist $i,j \in [n]$ with $i < j$ 
such that one of the following conditions holds: (1) $(a_i,b_i) \in
\ialphabet \times \{\epsilon\}$ and $b_j \in \ialphabet$, (2)
$(a_i,b_i) \in \{\epsilon\} \times \ialphabet$ and $a_j \in \ialphabet$.
Intuitively, as soon as the two heads go out of sync, the head that is lagging
behind is no longer allowed to move forward.
The relation $R \subseteq (\ialphabet^*)^2$ that
$\Aut$ \emph{recognises} consists of all tuples $\bar w$ for which there is a run 
\[
    \pi := q_0 \tran{\sigma_1} q_1 \tran{\sigma_2} \cdots \tran{\sigma_n} q_n
\]
of $\Aut$ (treated as an NFA) such that $\bar w = \sigma_1 \circ \sigma_2 
\circ \cdots \circ \sigma_n$, where the string concatenation operator $\circ$ 
is extended to tuples over words component-wise (i.e. $(v_1,v_2)
\circ (w_1,w_2) = (v_1w_1,v_2w_2)$). A relation is said
to be \defn{rational} if it is recognised by a transducer. 
A relation is said to be \defn{synchronised rational} (a.k.a. \defn{regular}
or \defn{automatic}; see \cite{BTV09,BG00,BFL13})
if it is recognised by a synchronised transducer.
Rational relations satisfy some nice properties (e.g. see
\cite{Berstel,Sakarovitch}): (1) closure under union and concatenation, and 
(2) the pre/post image of a regular language under a rational relation is 
regular. The transducers/automata witnessing the above two properties can also 
be constructed efficiently: taking union can be done in linear-time, while
taking concatenation and pre/post image of a regular language can be done
in quadratic time (e.g. see \cite{Berstel,BG07,BG08}).
Synchronised rational relations are adequate for certain applications (e.g. see 
\cite{BG00,BLLW12,TL10,yu2011}), while also satisfying effective
closure under intersection and complementation (cf. \cite{BG00}).

\begin{example}
    \label{ex:erase}
\em
The operator \texttt{replace-all} replaces all occurrences of subwords matched 
by a regular expression $e$ by a word in a regular expression $e'$, which in a 
Vim-like notation can be written as \texttt{s/$e$/$e'$/g}. There are various 
matching strategies that are used by real-world programming languages, e.g., 
first match, longest match, etc. They can all be encoded as transducers (e.g. 
see \cite{SMV12}). One particular use of \texttt{replace-all} is to replace 
words that match a regular language
$L$ by $\epsilon$ (i.e. an \emph{erase} operation). Such usage of 
\texttt{replace-all} can be found in sanitisation of PHP scripts,
e.g., see \cite{Balzarotti08,fmsd14,patching,FL10,FPBL13}. For example, to 
thwart XSS attack patterns of the form
\texttt{<script>}$\ialphabet^*$\texttt{</script>} from a string variable
$x$, one could erase each occurrence of \texttt{<} from $x$ (e.g. see
\cite{fmsd14,patching,Balzarotti08}). 

Let $y = \replace{x}{\epsilon/A}$ denote the operation of erasing each
occurrence of letters $a \in A$ (e.g. $A = \{\text{\texttt{<}}\}$) from 
$x$ and assign it to $y$. The transducer $\Transducer$ for this is simple. It 
has one state $q$, which is both an initial and final state. It has $|A|$ 
transitions, i.e.,
for each $a \in A$ the transducer 
$\Transducer$ has the transition $(q,(a,\epsilon),q)$.
\qed
\end{example}

\paragraph*{{\bf Computational complexity:}} 
In this paper, we study not only decidability but also the \emph{complexity}
of string logics. Pinpointing the precise complexity of
verification problems is not only of fundamental importance, but also it often
suggests algorithmic techniques that are most suitable for attacking the
problem in practice. In this paper, we deal with the following computational
complexity classes (see \cite{Sipser-book} for more details): 
$\P$ (problems solvable in polynomial-time), 
$\PSPACE$ (problems solvable in polynomial space and exponential time), and
$\EXPSPACE$ (problems solvable in exponential space and double exponential
time).
Verification problems that have complexity $\PSPACE$ or beyond --- see
\cite{schwoon-thesis,MONA} for a few examples --- have substantially benefited 
from techniques like symbolic model checking \cite{McMillan}. 
As we shall
see later, our complexity upper bound also suggests the maximum lengths of
words that need to be explored to \emph{guarantee completeness}.
 
\section{The Core Constraint Language}
\label{sec:lang}

We start by defining a general string constraint language that supports
concatenations, finite-state transducers, and regex matching. The
language is a natural generalisation of three decidable string constraint
languages: word equations, finite-state transducers, and regex matching. 
The generality of the language, however, quickly makes it undecidable. 
To delineate the border of undecidability, we shall show that the 
undecidability already holds for various restrictions that have been 
proposed in the literature (e.g. restricting finite-state transducers to 
\texttt{replace-all}). We will recover decidability in the next section.
%

\OMIT{
We start by introducing a general yet simple constraint language for
reasoning about strings. This language consists of {\em relational} and
{\em regular
constraints}, which allow to express three important
features of our setting: concatenation, transducers, and language
constraints.  As we shall see, the generality of the language makes it
easily undecidable. Later, in Sections \ref{sec:decidable_string} and
\ref{sec:ext}, we shall provide a 
decidable fragment of this language, as well as a decidable extension
of such fragment with integer and character constraints, that are 
sufficiently expressive for our applications.
}

\subsection{{\bf Language Definition}}
We assume a finite alphabet $\ialphabet$ and countably many 
{\em string variables} $x,y,z,\dots$ ranging over
$\Sigma^*$. We start by defining \defn{relational constraints}. 

\begin{definition}[Relational constraints]
An \emph{atomic relational constraint} over $\Sigma$ is an expression 
$\varphi$ defined by the following grammar:
\[
\varphi \ ::= \ y = x_1 \circ \cdots \circ x_n (n \in \N)\ |\ y = w\ |\ 
R(x,y)
\]
where $y, x_i$ are string variables, 
$w \in \ialphabet^*$ is a constant word,
and $R$ is a rational relation over 
$\Sigma$ given as a transducer.  
Here, $\circ$ is used to denote the string concatenation operator, which we
shall often omit (or simply replace by $\cdot$) to avoid notational clutter. 
A \emph{relational constraint} is a conjunction of atomic relational 
constraints. 
\end{definition} 

In other words, an atomic relational constraint allows us to test equality of a 
string variable $y$ with either a concatenation of string variables or a 
constant string, or whether the transducer $R$ can transform $x$ into $y$.
\anthonychanged{Notice that the atomic constraint $y = x_1 \circ \cdots \circ x_n$ cannot 
be defined as a rational relation $R(x,y)$ (or in fact any binary 
relation) when $n > 1$.}
We now define regular constraints (i.e. regex matching constraints), which 
check whether a word belongs to a boolean combination of regular languages. 

\begin{definition}[Regular constraints] An {\em atomic regular
    constraint} over $\Sigma$ is an expression of the form $P(x)$, for
  $P$ a regular language over $\Sigma$ given as an NFA and $x$ a
  variable. 
A \emph{regular constraint} over $\Sigma$ is a boolean combination of atomic 
regular constraints over $\Sigma$ defined by the following grammar:
\[
    \varphi \ ::= \ P(x)\ |\ \varphi \wedge \varphi\ |\ \varphi \vee \varphi\ 
        |\ \neg \varphi 
\]
\end{definition} 

We finally define the class of string constraints by taking the
conjunction of relational and regular constraints. 

\begin{definition}[String constraints]  
A \defn{string constraint} over finite alphabet 
$\Sigma$ is a conjunction of a relational
constraint and a regular constraint over $\Sigma$. 
\end{definition}

String constraints allow us to express general
    \defn{word equations} (i.e., $x_1\ldots x_n = y_1\ldots y_n$ for 
not necessarily distinct variables) by a conjunction of
$y = x_1\ldots x_n$ and $y = y_1\ldots y_n$. 
In addition, when the word equation asserts that one of the
$x_i$'s or $y_j$'s is the constant string $w \in \ialphabet^*$, we simply add a
regular constraint that enforces the variable to belong to the language 
$\{w\}$. 

An \defn{assignment} for a string 
constraint $\varphi$ over $\Sigma$ is simply a mapping $\iota$ from
the set of string variables mentioned in $\varphi$ to $\Sigma^*$.
 It \defn{satisfies} $\varphi$ if the constraint
$\varphi$ becomes true under the substitution of each variable $x$ by $\iota(x)$.
We formalise this for atomic relational and regular constraints (boolean
connectives are standard):
\begin{enumerate}
\item $\iota$ satisfies the relational constraint   
 $y = x_1 \cdots x_n$, for string variables $y,x_1,\dots,x_n$, if and
 only if $\iota(y) = \iota(x_1) \cdots \iota(x_n)$. 
\item $\iota$ satisfies the relational constraint   
 $y = w$, for string variable $y$ and word $w \in \Sigma^*$, if and
 only if $\iota(y) = w$. 
\item $\iota$ satisfies the relational constraint $R(x,y)$, for 
    a rational relation $R$, if and
  only if the pair $(\iota(x),\iota(y))$ belongs to $R$. 
\item $\iota$ satisfies the atomic regular constraint $P(x)$, for $P$
  a regular language, if and only if $\iota(x) \in P$.
\end{enumerate} 

A satisfying assignment for $\varphi$ is also called a \defn{solution} for 
$\varphi$. If $\varphi$ has a solution, then it is said to be 
\defn{satisfiable}.

\begin{example}
    \em
    In Introduction, we have expressed the program logic of the script
    in Example \ref{ex:cacm} and the attack pattern as a conjunction of
    four atomic relational constraints and one regular constraint. \qed
\end{example}

\OMIT{
\begin{remark}
A word of caution is necessary here: it is important to remember that $R$
is a relation, \emph{not} a function in general. So, although writing
$y = R(x)$ is more natural, the more proper way of writing it is 
$R(x,y)$. We should do so when we want to emphasise that $R$ is a relation.
\end{remark}
}

\subsection{{\bf The Satisfiability Problem}}

The definition of the problem is given as follows.

 \begin{center}
\fbox{\begin{tabular}{ll}
{\small PROBLEM} : & \sat. \\{\small INPUT} : & A string constraint $\varphi$
over $\Sigma$. \\
{\small QUESTION} : & Is $\varphi$ satisfiable?  
\end{tabular}}
\end{center}

The generality of the constraint language makes it undecidable, even
in very simple cases.

\begin{proposition}
\sat\ is undecidable. 
\end{proposition}

This is because checking satisfiability of constraints of the form
$R(x,x)$ is already undecidable by a simple reduction from the
\emph{Post Correspondence Problem} (PCP), e.g., see \cite[Proof of
Proposition 2.14]{Morvan00}. For this reason, an \emph{acyclicity}
constraint is often imposed (e.g. see \cite{BFL13,BG07,BG08}) to
obtain decidability for formulas that are conjunctions of constraints
of the form $P(x)$ or $R(x,y)$, where $P$ is a regular language and
$R$ is a rational constraint. This condition is defined as follows. Let $\varphi$ be a
formula of the form above and $\GraphR(\varphi)$ the 
undirected graph whose nodes are the variables in $\varphi$ and there is
an edge $\{x,y\}$ if $R(x,y)$ is a constraint in $\varphi$. 
Further, let $\sf AC$ be the class of those formulas $\varphi$ such
that $\GraphR(\varphi)$ is acyclic.
Then: 

\begin{proposition}[\cite{BFL13}]
    Checking satisfiability of formulas in $\sf AC$ is 
$\PSPACE$-complete. In fact, if a formula $\varphi \in \AC$ is satisfiable,
then it has a solution of size at most exponential in $|\varphi|$.
    \label{prop:BFL-PSPACE}
\end{proposition}

The authors of \cite{BFL13} refer to this problem as \emph{the generalised 
intersection problem with acyclic queries} (see Theorem 6.7 in \cite{BFL13}). 
[Decidability (in fact, in exponential time) of such a
restriction already follows from 
the classic result by Nivat in the study of rational
relations (e.g. see the textbook \cite{Berstel}) that the pre/post
images of regular languages under a rational transducer is effectively
regular and the complexity analysis in \cite{BG07,BG08}.]

\OMIT{
\begin{remark} \label{rm:nivat} 
Decidability (in fact, in exponential time) of such a
restriction already follows from an old result. It can be obtained by
first using the classic result by Nivat in the study of rational
relations (e.g. see the textbook \cite{Berstel}) that the pre/post
images of regular languages under a rational transducer is effectively
regular. In fact, Nivat's construction outputs an NFA for the language
$\{ x \in \Sigma^* : \exists y( R(x,y) \wedge P(x) \}$, for a
transducer $R$ and an NFA $P$, in time $O(|R|\cdot |P|)$ (see, e.g.,
\cite{BG07,BG08} for more details). Next, one applies the product
automata construction to construct an NFA recognising the intersection
of regular languages at each point of branching in $\GraphR(\varphi)$.
This NFA is of exponential size, and thus can be checked for non-emptiness in 
exponential time. 
\end{remark}
}

Unfortunately, the positive result in Proposition
\ref{prop:BFL-PSPACE} 
cannot be easily extended to our
constraint language, as taking the conjunction 
of a formula in $\sf AC$ (in particular, of a single rational
constraint $R(x,y)$) with the very simple word
equation $x = y$ turns the satisfiability problem undecidable. This is because
the undecidable constraint $R(x,x)$ can be expressed as: $x = y \, \wedge \, 
R(x,y).$
Restricting $R$ to synchronised rational relations 
does not help either: satisfiability of string constraints of the form
    $x = yz \, \wedge \, R(x,z),$
where $R$ is a synchronised rational relation, is undecidable
\cite{BFL13}. 

Another option is to restrict the use of finite-state transducers to
the \texttt{replace-all} operators. As we have argued in Example
\ref{ex:erase} this might be sufficient to model some sanitisation 
functions that arise in practice. In fact, some string constraint languages
that have been proposed in the literature (e.g. see
\cite{S3,fmsd14}) permit the use of the \texttt{replace-all} operator, but
not finite-state transductions in general. It turns out that this restriction
is still undecidable even for the very restricted use of \texttt{replace-all}
of the form $\replace{x}{\epsilon/A}$ (defined in Example \ref{ex:erase}) which
erase all occurrences of characters $a \in A$ in $x$. The proof (see
\shortlong{full version}{Appendix}) is via a simple but tedious reduction
from PCP.
\OMIT{
Moreover, we prove next that an undecidability
result also holds when we restrict rational relations to the \texttt{replace-all} 
operation that simply erases certain characters (as in the program in Figure 
\ref{prog:undec}).
}

\begin{proposition}
    Checking satisfiability of a constraint of the form $x = yz \wedge
\varphi$, where $\varphi$ is a formula in $\sf AC$ that only mentions 
transducers $R(x,y)$ of the form $\replace{x}{\epsilon/A}$, 
    is undecidable.
    %
    \label{prop:undec-prog}
\end{proposition}

\section{The Straight-Line Fragment}
\label{sec:decidable_string}

In Section \ref{sec:lang}, we have explored various syntactic restrictions of 
the
core constraint language that still permit both concatenation and transducers,
and saw that undecidability still held.
In this section, we will show that the ``straight-line fragment'' of the
language is decidable (in fact, solvable in exponential space). 
This straight-line fragment captures the structure of straight-line 
manipulating programs with concatenations and finite-state transductions as
atomic string operations. Note that straight-line programs naturally arise when 
verifying string-manipulating programs using bounded model checking and
dynamic symbolic executions (e.g. see \cite{S3,KS08,DKW08}),
which
unrolls loops in the programs (up to a given depth) and converts the resulting 
programs into \emph{static single assignment form} (i.e. each variable defined 
only once). For applications to detecting mutation XSS \cite{mXSS}, we will
see that the formula for analysing mutation XSS from Example \ref{ex:cacm} 
in Introduction can be expressed in the straight-line fragment. We will also
see other such examples in this section.


\begin{convention}
In the sequel, we will treat an atomic relational constraint of 
the form $y = w$, for a word $w \in \Sigma^*$, as an atomic regular 
constraint, i.e., simply treat $w$ as regular expression and assert $y \in 
L(w)$. 
\end{convention}

To define the straight-line fragment of the core constraint language, we
first write $R(x,y)$ as $y = R(x)$. This notation is quite natural since
$R$ can be viewed as a string transformation from the input $x$ to the output 
$y$.
[However, a word of caution is necessary: it is important to remember that $R$ 
    is a relation that need not be a function in general.]  
    We also say that a variable $x$ is \defn{a source variable} in the 
    relational 
    constraint $\varphi$ if there is no conjunct in $\varphi$ of the form
    $x = x_1\cdots x_n$ or $x = R(y)$ for some transducer $R$.

\begin{definition}[Straight-line constraints] \label{defn:SL}
    A relational constraint $\varphi$ 
    is said to be \defn{straight-line} if it can be rewritten
    (by reordering the conjuncts) as a relational constraint of the form
    $\bigwedge_{i=1}^m x_i = P_i$ such that:
    \begin{description}
        \item[(SL1)] $x_1,\ldots,x_m$ are different variables
        \item[(SL2)] Each $P_i$ uses only source variables in $\varphi$ or 
            variables from $\{x_1,\ldots,x_{i-1}\}$. 
    \end{description}
    We say that a 
    string constraint is straight-line if it is a conjunction of a straight-line 
relational
     constraint with a regular constraint. Let $\AUWR$ be the set of
    all straight-line string constraints.
\end{definition}

An example of a straight-line string constraint is $y = R(x) \wedge z = yyz'$. Another
example is the constraint from Example \ref{ex:cacm}. Straight-line
restrictions also rule out
all the undecidable constraints from Section \ref{sec:lang}, e.g., formulas of the form
$x = y \wedge R(x,y)$ and $x =yz \wedge R(z,x)$. 

The literal definition of straight-line constraints does not give an
efficient algorithm for checking whether a constraint is straight-line.
This, however, can be done efficiently.

\begin{proposition}
    There is a linear-time algorithm for checking whether a relational
    constraint $\varphi$ can be made a straight-line constraint (by reordering 
    equations) and, if so, outputs the reordered constraint 
    $\bigwedge_{i=1}^m x_i = P_i$ satisfying \textbf{(SL1)} and
    \textbf{(SL2)}.
    \label{prop:linear}
\end{proposition}

The proof of this proposition is standard. Straight-line relational constraints can be 
visualised by drawing a ``dependency
graph'' of the variables in the constraints. Formally, given a relational 
constraint $\varphi$, 
the \defn{dependency graph} $\G(\varphi)$ of $\varphi$ is the \emph{directed} 
graph whose nodes are the string 
variables appearing in $\varphi$ and there is an edge from variable $x$ to $y$ 
iff (a) $\varphi$ contains a conjunct of the form $R(x,y)$, for a rational 
relation $R$, or (b) an equation of the form $y = x_1\ldots x_n$ for some 
string variables $x_1\ldots x_n$ which include $x$. It is easy to see that
straight-line constraints have acyclic dependency graphs. In fact,
the converse is also true provided that the relational constraint $\varphi$ is
\defn{uniquely definitional} (i.e. there are no two conjuncts $x = P$ and 
$x= P'$ in $\varphi$ with the same left-hand side variable). A linear-time algorithm
for Proposition \ref{prop:linear}, then, first checks if the given 
constraint $\varphi$ is uniquely definitional by sequentially going through
each conjunct while maintaining $m$ bits in memory (one for each
variable in the left-hand side of an equation). Once 
unique definitionality has been checked, the algorithm checks whether the
directed
graph $\G(\varphi)$ is acyclic and, if so, output a topological sort, which
is well-known to be solvable in linear-time (e.g. see \cite{Cormen}).
The topological sort corresponds to a reordering of the conjuncts in
$\varphi$ so that \textbf{(SL1)} and \textbf{(SL2)} are both satisfied.

Our main result states that satisfiability for $\AUWR$ is decidable.

\begin{theorem} \label{th:expspace}  
\sat\ for the class $\AUWR$ is $\EXPSPACE$-complete. 
\end{theorem} 

Recall that $\EXPSPACE$ problems require double-exponential time
algorithms in the worst case. 
An important 
corollary of the proof of Theorem \ref{th:expspace} is a bounded model property.

\begin{theorem} \label{theo:smp}
If a string constraint $\varphi$ in $\AUWR$ is satisfiable, then it has a solution with each
word of length at most $2^{2^{p(|\varphi|)}}$, for some
polynomial $p(x)$.
\end{theorem}
\anthonychanged{
In Theorem \ref{th:pspace} and Theorem \ref{th:pspace-bmp} below, we identify a 
natural restriction of $\AUWR$ that yields an upper bound for satisfiability 
--- $\PSPACE$ (i.e., a single-exponential time) and satisfying models of size
at most single-exponential --- and seems sufficiently
expressive in practice (e.g. it subsumes all our examples)}. 

\begin{remark}
    \anthonychanged{
    The reader might be wondering whether
    Theorem \ref{th:expspace} and Theorem \ref{theo:smp} immediately follow from
    Proposition \ref{prop:BFL-PSPACE}? This is not the case
    since $\AC$ does not permit equalities and concatenations (therefore, 
    expressions of the form $x = y.z$ cannot be expressed). 
    Similarly, the theorems also do not immediately follow from the
    effective closure of regular languages under (1) pre/post images under
    rational relations (see discussion below
    Proposition \ref{prop:BFL-PSPACE}), and (2) concatenation. For example,
    consider the constraint $x = yy \wedge y \in L(a^*+b^*) \wedge x \in L(ab)$
    over the alphabet $\ialphabet = \{a,b\}$. This is unsatisfiable. 
    Instead, naively applying the two aforementioned closure, we would deduce 
    that $x$ matches $(a^*+b^*).(a^*+b^*)$, which can be matched by $ab$, i.e.,
    a false positive. As we shall see later, multiple occurrences of variables 
    in an assignment yield constraints with ``dimensions'' $> 1$ (definition
    below), instances of which include all the mutation XSS examples in the 
    paper.
}
\end{remark}




\subsection{Application to Detecting Mutation XSS}
Before proving Theorems \ref{th:expspace} and Theorem \ref{theo:smp}, 
we will mention that the logic
$\AUWR$ is sufficiently powerful for expressing string constraints that arise
from analysis of mutation XSS vulnerability in web applications. By Theorem
\ref{th:expspace}, such a vulnerability analysis can be performed automatically.
We have seen one such example (i.e. Example \ref{ex:cacm}).
%
We will now provide other examples of mutation XSS attacks
that can be expressed within the framework of $\AUWR$.
\begin{example}
    \em
    We have seen that the script in Example \ref{ex:cacm} contains an XSS
    bug. As suggested in \cite{Kern14}, the corrected version of the script
    swaps the \emph{order} of the sanitisation functions \texttt{escapeString} 
    and \texttt{htmlEscape} resulting in the following script:
\begin{flushleft}
    \footnotesize
    \begin{verbatim}
var x = goog.string.escapeString(cat);
var y = goog.string.htmlEscape(x);
catElem.innerHTML = '<button onclick=
    "createCatList(\'' + y + '\')">' + x + '</button>';
    \end{verbatim}
\end{flushleft}
This script is no longer vulnerable to the attack pattern \texttt{e1} provided
in Example \ref{ex:cacm}. The program logic of the corrected script can be 
expressed in $\AUWR$ using the same formula as for Example \ref{ex:cacm} (see
Introduction) except that the transducers $R_1$ and $R_2$ are swapped. The
algorithm from Theorem \ref{th:expspace} will be able to automatically
point out that the script is secure against the attack pattern \texttt{e1}.
\qed
\label{ex:corrected}
\end{example}

\begin{example}
    \label{ex:mxss1}
    \em
    This example is an adaptation of vulnerable code patterns from
    \cite[Listing 1.12]{mXSS} applied to the previous example. After the 
    JavaScript from Example \ref{ex:cacm}
    has been corrected in Example \ref{ex:corrected}, suppose now that a 
    programmer wishes to introduce a new ``title'' HTML element into
    the HTML document 
    in which the new catalogue category name will be
    displayed. The following code snippet contains two lines in this HTML file:
\begin{flushleft}
    \footnotesize
    \begin{verbatim}
<h1>New catalogue category: 
    <span id="node1">TBA</span></h1>
<div id="node2"></div>
    \end{verbatim}
\end{flushleft}
    The following JavaScript code snippet is a modification of the JavaScript 
    from
    Example \ref{ex:corrected} which additionally puts the new catalogue 
    category name in the title (i.e. ID \texttt{node1}):
\begin{flushleft}
    \footnotesize
    \begin{verbatim}
var titleElem = document.getElementById("node1");
var catElem = document.getElementById("node2");

var x = goog.string.escapeString(cat);
titleElem.innerHTML = x;
var y = goog.string.htmlEscape(titleElem.innerHTML);
catElem.innerHTML = '<button onclick=
    "createCatList(\'' + y + '\')">' + x + '</button>';
    \end{verbatim}
\end{flushleft}
    This JavaScript now contains a subtle mXSS vulnerability. Consider
    the value \verb+cat = '&#39;);alert(1);//'+. The value \verb+x+
    will be the same as \verb+cat+ since the metacharacters \verb+'+ and
    \verb+"+ do not occur in \verb+cat+. The value of 
    \texttt{titleElem.innerHTML} is, however, \verb+');alert(1);//+ since
    an implicit browser transduction occurs upon writing into the DOM via
    \texttt{innerHTML}. HTMLescaping \verb+');alert(1);//+ results in the 
    string \verb+&#39;);alert(1);//+, which is the value of \verb+y+.
    Another implicit browser transduction takes place when assigning a
    value to \texttt{catElem.innerHTML}.  The DOM element \texttt{catElem} now 
    contains the HTML markup:
\begin{flushleft}
    \footnotesize
    \begin{verbatim}
<button onclick="createCatList('');alert(1);//')">
    ');alert(1);//')</button>
    \end{verbatim}
\end{flushleft}
    Upon clicking the button, the browser executes the attacker's
    script \texttt{alert(1)} after \verb+createCatList+ has been invoked.
    The XSS bug is due to the programmer's wrong assumption that
    the value of \texttt{titleElem.innerHTML} is the same as
    \texttt{x} after the assignment. 

    Let us now encode the program logic of the above JavaScript 
    as a straight-line string constraint. Let \texttt{e1} be the attack
    pattern from Example \ref{ex:cacm}. As in Example \ref{ex:cacm},
    let $R_1$ and $R_2$ be transducers implementing \texttt{htmlEscape}
    and \texttt{escapeString}, and let $R_3$ be the transducer
    implementing the implicit browser transductions upon \texttt{innerHTML}
    assignments. The desired formula can now be written as a conjunction of
    \begin{itemize}
    \item $x = R_2(\text{\texttt{cat}})$
    \item $\text{\texttt{titleElem.innerHTML}} = R_3(x)$
    \item $y = R_1(\text{\texttt{titleElem.innerHTML}})$
    \item $z = w_1 \cdot y \cdot w_2 \cdot x \cdot
        w_3$ for some constant strings $w_1, w_2, w_3$
    \item $\text{\texttt{catElem.innerHTML}} = R_3(z)$
    \item $\text{\texttt{catElem.innerHTML}}$ matches \text{\texttt{e1}}.
    \end{itemize}
    Observe that this is a straight-line string constraint. Therefore,
    the algorithm from Theorem \ref{th:expspace} will be able to
    automatically detect a vulnerability against the attack pattern
    \texttt{e1}.
    \qed
\end{example}

In the \shortlong{full version}{appendix}, we will provide another example of 
mXSS bug from adapted from \cite{Stock14} and show how it can be analysed 
within the framework of $\AUWR$. 

\begin{remark}
    At this stage, the reader might be wondering about the security 
    implication when the above
    formulas $\varphi$ are satisfiable or unsatisfiable. Unsatisfiability
    rules out specific vulnerability pattern. In the case of satisfiability,
    since attack patterns
    (including \texttt{e1}) generally \emph{overapproximate} a set of
    bad strings, a solution to $\varphi$ might not correspond to an actual
    attack. 
    There are multiple proposals to address this problem (e.g. see
    \cite{fmsd14,Balzarotti08}). One method (e.g. see \cite{Balzarotti08})
    is to supply each attack pattern (e.g. 
    $\Sigma^*\text{\texttt{<script>}}\Sigma^*$) with a set of test cases in
    the form of actual strings
    (e.g. \texttt{<script>alert(1);</script>}). Replacing \texttt{e1}
    by a specific test case, $\varphi$ can be checked again for
    satisfiability.
\end{remark}

\subsection{{\bf Proofs of Theorem \ref{th:expspace}} and Theorem 
\ref{theo:smp}}
\label{sec:ub} 

We start by proving the upper bound of Theorem \ref{th:expspace}. Then
we explain how Theorem \ref{theo:smp} follows from it. Finally, we
sketch the proof of the lower bound of Theorem \ref{th:expspace}. 

\subsubsection{Upper Bound of Theorem \ref{th:expspace}} 

Let $\AUR$ denote the restriction of $\AUWR$ to formulas which do not involve 
any concatenation, i.e., constraints of the form $y = y_1\ldots y_n$.
The crux of the algorithm witnessing Theorem \ref{th:expspace} is that we 
transform the input constraint $\varphi$ in $\AUWR$ into one in the class $\AUR$. 
After this, in Step 3 we will apply a classic result in the theory of rational 
relations to
obtain decidability (or a recent result \cite{BFL13} for a better complexity).
We now provide the details of our algorithm for the satisfiability problem for 
$\AUWR$.

\paragraph*{Step 1: simplification of regular constraints}
Recall that a regular constraint is a boolean combination of constraints of the
form $P(x)$, where $P$ is a regular language given as an NFA over a
finite alphabet $\Sigma$. Given a regular
constraint $\psi$, there exists an equivalent constraint in disjunctive normal
form (DNF) of exponential size, where each disjunct $\theta$ is a conjunction 
of literals involving the atoms from $\psi$ of the form $P(x)$ (so $\theta$ has 
size linear in $|\psi|$).
From propositional logic, we know that there is a standard enumeration of these
disjuncts, i.e., enumerate satisfying assignments of $\psi$ treated as a 
propositional formula (i.e. each atom $P(x)$ is now a proposition). Such an
enumeration runs in polynomial space and exponential time.

Next, each disjunct $\theta$ in the aforementioned enumeration can now be converted 
into a conjunction 
\begin{equation} \label{eq:prc} 
    P_1(x_1) \wedge \cdots \wedge P_n(x_n)
\end{equation} 
of positive literals, where each string variable $x_i$ is constrained only by
precisely one atomic regular constraint $P_i$. This can be done by complementing
each NFA that occurs as a negative literal in $\theta$, and by computing a
single NFA for the intersection of regular languages by a standard product automata
construction in the case when a string variable is constrained by several literals
in $\theta$. This construction runs in exponential
time. 


Let $\varphi$ be the given constraint in $\AUWR$, consisting of the conjunction
of a relational constraint $\chi$ and a regular
constraint $\psi$. In order for $\varphi$ to 
be satisfiable, it is sufficient and necessary that $\chi \wedge \theta$ 
is satisfiable, for some disjunct $\theta$ in the enumeration of disjuncts of $\psi$ 
in DNF. 
In this step, our algorithm simply guesses such disjunct
$\theta$ (which is of polynomial size). The remaining steps of the
algorithm then check whether $\chi \wedge \theta$ is satisfiable.   
From our previous remarks, 
we may assume then that $\theta$ 
is of the form \eqref{eq:prc}, where each string
variable $x_i$ is uniquely constrained by a literal
$P_i(x_i)$. We also assume, without loss of generality, that each
variable $y$ in $\chi$ is constrained by some literal $P(y)$ in
$\theta$. Otherwise, we simply add to $\theta$ a literal which states that $y$
belongs to $\Sigma^*$. 

\paragraph*{Step 2: Removing concatenation}
For this step, we will transform a given constraint 
$\varphi$ in $\AUWR$ into a constraint $\varphi'$ that uses only atomic string 
constraints
of the form $x = R(y)$, i.e., all string constraints of the form $x=y_1\ldots y_n$
are removed. Since word equations of the form $w = yy$ cannot in general
be expressed as a transducer, our transformation cannot possibly express the same
property that $\varphi$ expresses, i.e., it is impossible that $\varphi$ and 
$\varphi'$ have the same set of satisfying assignments in general. However, as we 
shall see later, 
\emph{by introducing extra variables and allowing both conjunctions/disjunctions for
our string constraints} it is possible to produce the formula $\varphi'$ without
concatenation operators
whose satisfying assignments can be easily transformed into satisfying assignments 
of $\varphi$ and vice versa (see Lemma \ref{lm:reduce-correctness} below).

\OMIT{
involving only the \emph{source variables} 
(i.e. nodes in $\G(\varphi)$ with no incoming edges) and the 
\emph{sink variables} (i.e. nodes in $\G(\varphi)$ with no outgoing edges).
This algorithm can be construed as computing a ``normal form'' for $\varphi$. 

\begin{definition}
    Suppose that $\varphi$ has source variables $x_1,\ldots,x_n$ and 
    transducers $R_1,\ldots,R_m$. A \defn{term} over $\varphi$ is defined by 
    induction: (1) each $x_i$ is a term over $\varphi$, and (2) if $t$ is a 
    term over $\varphi$ and $R = R_i$ for some $i \in [1,m]$, then 
    $R_{[q,q']}(t)$ for all states $q, q'$ of $R$.
\end{definition}
Note that each term has a \defn{unique} string variable.

\begin{definition}
    We define an intermediate language $\AUWRe$ of string constraints. 
    An atomic string constraint in $\AUWRe$ is of the form $y = t_1\ldots t_n$, 
    where
    $t_1,\ldots,t_n$ are terms (over some string constraint $\varphi$ in $\AUWR$).
    Each string constraint in $\AUWRe$ is of the form 
\end{definition}
}

The structure of the straight-line relational constraint $\varphi$ immediately gives us 
a topological sort $x_1 \prec \ldots \prec x_m$ on $\G(\varphi)$. 
We may assume without loss of generality that all the source nodes are at the
beginning of the topological sort, i.e., it is not the case that $x_i \prec 
x_{i+1}$ for some non-source node $x_i$ and some source node $x_{i+1}$ 
(which might
happen if they are incomparable in $\G(\varphi)$ construed as a partial order). 
Next, for each variable $y$ in $\varphi$, we will define a relational 
constraint
$\NewStrCon{y}$ involving only rational relations (but with conjunctions and
disjunctions), a regular constraint $\NewRegCon{y}$, and fresh
variables $y(1),\dots,y(\max_y)$, where $\max_y$ is a positive
integer. We will define these
by induction on the position of $y$ with respect to the order $\prec$.

\smallskip 

\underline{Base cases}: A source node $y$ with a regular constraint $P(y)$. For this, 
we set 
$\NewStrCon{y} := \top$ and $\NewRegCon{y} := P(y)$. We also set
$\max_y = 1$ and define a fresh variable $y(1)$.

\smallskip 

\underline{First inductive case}: A non-source node $y$ with an assignment $y = 
y_1\ldots y_n$ and a regular constraint $P(y)$, for some $y_1,\ldots,y_n \prec y$. 
By induction, we assume 
that $\max_{y_i} \in \Z_{>0}$ and new string variables $y_i(1),\ldots,y_i(\max_{y_i})$
have all been defined (for each $1 \leq i \leq n$).
The main idea behind our construction 
is to interpret each assignment for $y_i$ ($1 \leq i \leq
n$) as the
``concatenation'' of an assignment for   
$y_i(1),\ldots,y_i(\max_{y_i})$, respectively.  

Formally, let $\max_y = \sum_{j=1}^n 
\max_{y_j}$ and $S = \{1,\dots,\max_y\}$. We define a {\em selector
function} 
$\nu: S \to \Z_{>0}^2$ as follows: For each $k \in S$ it is the case
that $\nu(k)$ is a pair of numbers $(i,l)$, where $i$ is
the smallest number $i$ such that $k \leq \sum_{j=1}^i \max_{y_j}$ and
$l := k - \sum_{j=1}^{i-1} \max_{y_j}$. Intuitively, 
if $\nu(k) = (i,l)$, then
$\nu$ ``selects'' the variable $y_i$ for the fresh variable $y(k)$, 
and $l$ further ``refines'' this selection to $y_i(l)$. 

In the set $\NewStrCon{y}$ we ``define'' the fresh variables $y(k)$,
for each $k \in S$. We do so by setting $\NewStrCon{y}$ to be the 
conjunction of all formulas of the form $y(k) = y_i(l)$, where 
$k \in S$ and $(i,l) = \nu(k)$. Notice that we can express each such conjunct as a 
transducer $R(y_i(l),y(k))$, where $R$ is the identity transducer,
i.e., the one that outputs
its input without modification. 

We now define $\NewRegCon{y}$. To this end, we first
observe that an accepting run of the automaton $P$ on a word $w = w_1\cdots w_N$,
can be split into $N$ subruns. To make this notion more precise, we define an
\defn{$N$-splitting} $\sigma$ of an automaton $\Aut$ with states $\controls$ to be
a sequence $q_0, \ldots, q_{N} \in \controls$ such that each state $q_i$ ($i \in 
\{1, \ldots,N\}$) is reachable from the state $q_{i-1}$ in $\Aut$.
Recall that $\Aut_{[q,q']}$ is the NFA $\Aut$ whose initial (resp. set of final
states) is replaced by $q$ (resp. $\{q'\}$). Then
$\NewRegCon{y}$ is defined as: 
\[
    \bigvee_{ \sigma = q_0,\ldots,q_{\max_y}}
    P_{[q_0,q_1]}(y(1)) \, \wedge \, \cdots \, \wedge \,  
    P_{[q_{\max_y-1},q_{\max_y}]}(y({\max}_y)), 
\]
where $\sigma$ ranges over $\max_y$-splittings of $P$. That is, 
$\NewRegCon{y}$ states that there is some $\max_y$-splitting $q_0,\ldots,q_{\max_y}$ 
of $P$ such that each $y(k)$ ($1 \leq k \leq \max_y$) is accepted by $P_{[q_{k-1},q_k]}$. 
This amounts to $y$ being accepted by $P$ (since $y$ is interpreted as the 
``concatenation'' of    
$y(1),\ldots,y(\max_{y})$, respectively).

\smallskip 

\underline{Second inductive case}: \, A non-source node $y$ with an assignment $y = 
R(x)$ and a regular constraint $P(y)$, for some $x \prec y$. By induction, we assume 
that $\max_x \in \Z_{>0}$ and new string variables $x(1),\ldots,x(\max_x)$
have all been defined. The crux of the construction is to replace $y =
R(x)$ by ``splitting'' it into $\max_x$ different constraints. This is achieved by
splitting the transducer $R$ (syntactically seen as an NFA over $\Sigma_\e \times
\Sigma_e$). 
More precisely, let $\max_y = \max_x$. Then $\NewStrCon{y}$ is: 
\[
    \bigvee_{\sigma = p_0,\ldots,p_{{\max}_y}}
    \bigwedge_{i=1}^{{\max}_y} y(i) = R_{[p_{i-1},p_i]}(x(i)), 
\]
where $\sigma$ ranges over all ${\max}_y$-splittings of $R$.
We may define $\NewRegCon{y}$ precisely in the same way as in the first inductive 
case.

\smallskip 

\underline{The output formula $\varphi'$:} \, 
We have now defined $\NewStrCon{y}$ and $\NewRegCon{y}$
for each node $y$ in $\G(\varphi)$. The output of our transformation is the
formula $$\varphi' \ := \ \bigwedge_{y} (\NewStrCon{y} \wedge
\NewRegCon{y}),$$ 
where
$y$ ranges over all variables in $\varphi$. 

Notice that $\varphi'$ is not necessarily a string constraint, as
$\NewStrCon{y}$ might be a {\em disjunction} of relational constraints for
some $y$'s. The notion of satisfiability extends to this class of
formulas in the standard way. In particular,
an assignment $\iota$ for the variables in $\NewStrCon{y}$ satisfies
this formula iff it satisfies one of its disjuncts. The formula   
$\varphi'$ is satisfiable iff there is an
assignment $\iota$ for the variables of $\varphi'$ that satisfies $\NewStrCon{y} \wedge
\NewRegCon{y}$ for each string variable $y$ in $\varphi$.  

\smallskip 

\underline{Correctness:} \, 
The following lemma shows that
our transformation is satisfiability preserving; in fact, there is an easy way
to obtain satisfying assignments for $\varphi$ from those of $\varphi'$ and vice
versa. 

\begin{lemma}[Correctness] \label{lm:reduce-correctness}
    $\varphi$ is satisfiable iff $\varphi'$ is satisfiable.
\end{lemma}

The proof is done by induction on the position of nodes in the topological sort
$\prec$ of $\G(\varphi)$. To this end, given a node $y$ in $\G(\varphi)$, we define
$\varphi_y$ to be the set of conjuncts in $\varphi$ involving only $y$
and variables
that precede $y$ in $\prec$. That is, the conjuncts in $\varphi_y$ are
the ones that only mention variables in the set
$Z_y := \{ x : x \preceq y \}$. For example, if $\varphi$ contains the
constraint $z = yy$, then this constraint cannot
be a conjunct of $\varphi_y$ (but it is a conjunct of $\varphi_z$). 
Similarly, we define $\varphi_y' := 
\bigwedge_{x \in Z_y} \NewStrCon{x} \wedge \NewRegCon{x}$. 
That is, $\varphi_y'$ is the set of conjuncts in $\varphi'$ of the
form $\NewStrCon{x} \wedge \NewRegCon{x}$ involving
only variables in the set $Z_y' := \{ x(i) : x \in Z_y, \, 1 \leq i \leq
\max_x\}$.   
We will prove the following technical lemma, which is a stronger formulation of 
Lemma \ref{lm:reduce-correctness}.

\begin{lemma} \label{lm:correctness-hypothesis}
    For each node $y$ in $\G(\varphi)$ and every assignment 
    $\iota: Z_y \to \Sigma^*$ of variables in $\varphi_y$,
    the following are equivalent:
    \begin{description}
    \item[(i)] $\iota$ is a satisfying assignment for the formula $\varphi_y$.
    \item[(ii)] There is a satisfying assignment $\iota': Z_y' \to \Sigma^*$
        for $\varphi_y'$ such that
        $\iota(x) = \iota'(x(1))\circ \cdots\circ \iota'(x({\max}_x))$
        for each $x \preceq y$.
    \end{description}
\end{lemma}

The proof of Lemma \ref{lm:correctness-hypothesis} is by 
induction on the position of $y$ in the topological sort $\prec$ of
$\G(\varphi)$. Due to lack of space we relegate this proof to the
\shortlong{full version}{appendix}. 
Since Lemma 
\ref{lm:correctness-hypothesis} is a stronger formulation of Lemma 
\ref{lm:reduce-correctness}, the correctness of our construction in Step 2 
follows.

\paragraph*{Step 3: Solving the final formula}
After applying the transformation from Step 2, 
the size of the resulting formula $\varphi'$ 
could be exponential in $|\varphi|$ due to repeated applications of constraints of 
the form $y = y_1\ldots y_n$, where some variable $y_i$ occurs several times on the 
right-hand side of the equation. 
In particular, there are exponentially many conjuncts of the
form $\NewStrCon{x} \wedge \NewRegCon{x}$  in $\varphi'$, 
More precisely, the resulting formula $\varphi'$ is a conjunction of multiple
formulas of two types:
\begin{itemize}
    \item conjunctions of atomic formulas of the form $R(x,y)$, for some transducer $R$, or a 
            disjunction of several such formulas.
    \item atomic formulas of the form $P(x)$, for some regular language $P$, or a 
        disjunction of several such formulas.
\end{itemize}
So, except for the disjunctions, the formula $\varphi'$ satisfies the 
shape of the fragment $\AUR$ of $\AUWR$. In fact, it is not difficult to
remove these disjunctions without too much additional computational overhead.
Recall that disjunctions were caused by splitting automata or transducers.
Although in this case a conjunct can have exponentially many disjuncts, we may 
simply 
nondeterministically guess one of the disjuncts (in effect, guessing one of the 
splittings of the automata/transducers). 
Nondeterministic algorithms can be determinised at the cost 
of quadratically extra space \cite{Savitch70}. The resulting formula 
$\varphi''$ is now a conjunction of atomic
formulas of the form $R(x,y)$ or $P(x)$. Moreover, 
it is easy to prove that the undirected graph $\GraphR(\varphi'')$
defined in Section \ref{sec:lang} is acyclic: 

\begin{lemma} \label{lemma:ac} 
$\GraphR(\varphi'')$ is acyclic. Thus, $\varphi''
\in {\sf AC}$. 
\end{lemma} 

\begin{proof}
By construction of $\varphi' = \bigwedge_{y} 
(\NewStrCon{y} \wedge \NewRegCon{y})$, for each variable $y$ in
$\varphi'$ there is at most one $x$ such that $\NewStrCon{y}$ (and,
therefore, $\varphi'$) contains an atom of the form $y = R(x)$, for $R$
a rational transducer. From this it is clear that $\GraphR(\varphi'')$
is acyclic.   
\end{proof} 

Decidability in 
$\EXPSPACE$ can now be easily obtained by Proposition \ref{prop:BFL-PSPACE}, i.e.,
we apply Proposition \ref{prop:BFL-PSPACE}
on the formula $\varphi''$ of size at most exponential in $|\varphi|$ and so 
our resulting
algorithm runs in exponential space, as desired.

\subsubsection{Proof of Theorem \ref{theo:smp}}

\anthonychanged{
Suppose that $\varphi \in \AUWR$ is satisfiable.
The previous algorithm computes a formula $\varphi''$ in $\sf AC$ such that
$\varphi$ is satisfiable iff $\varphi''$ is satisfiable (cf. Lemma
\ref{lm:reduce-correctness}). This implies that $\varphi''$ is satisfiable
and, so by the bounded model property from Proposition \ref{prop:BFL-PSPACE},
$\varphi''$ has a solution of size at most exponential in $|\varphi''|$. 
Now, the size $|\varphi''|$ is $O(\Dim(\varphi) \times |\varphi|)$, where
$\Dim(\varphi)$ is the \defn{dimension} of $\varphi$ which 
is defined to be the maximum $\max_x$ over all string variables $x$ in 
$\varphi$. The value $\Dim(\varphi)$ is at most exponential in $|\varphi|$.
This implies that 
$\varphi''$ has a solution of size at most $F(|\varphi|) := 
2^{(\Dim(\varphi) \times 
|\varphi|)^{O(1)}} = 2^{2^{|\varphi|^{O(1)}}}$. Then, by Lemma 
\ref{lm:correctness-hypothesis}, 
$\varphi$ has a solution of size
$\Dim(\varphi) \times F(|\varphi|) = 
2^{(\Dim(\varphi) \times |\varphi|)^{O(1)}} = 2^{2^{|\varphi|^{O(1)}}}$,
giving us the desired
upper bound on the maximum solution size for $\varphi$ that we need to
explore.}

\subsubsection{Lower Bound of Theorem \ref{th:expspace}}

In order to prove that checking satisfiability of string 
constraints in $\AUWR$ is
$\EXPSPACE$-hard, we reduce from the acceptance problem for a
deterministic Turing machine $\M$ that works in space $2^{cn}$, for $c
> 1$. That is, we provide a polynomial time reduction that, given
an input $w$ to $\M$, it constructs a constraint $\varphi(w)$ in
$\AUWR$ such
that $w$ is accepted by $\M$ if and only if $\varphi(w)$ is
satisfiable. Due to lack of space, the complete proof is relegated to the
\shortlong{full version}{appendix}. We only sketch the main ideas below. 

The reduction starts by constructing a regular constraint of the form
$P_1(x) \wedge \dots \wedge P_m(x)$, in such a way that a word $w$
satisfies this constraint if and only if it codifies a sequence of
configurations of $\M$, the first such configuration corresponds to
the initial configuration of $\M$ on input $w$, and the final such
configuration corresponds to a final configuration of $\M$. The we
only require to check that each non-initial configuration in $w$ is
obtained from its preceding configuration by applying the transition
function of $\M$. This is done by adding a set of relational
constraints that creates an exponential number of copies of $x$ (with
equations of the form $y = xx$), but deletes certain distinguished
parts of each such copy (with suitable transducers).

\subsection{{\bf A $\PSPACE$ Restriction}}
\OMIT{
We will finish this section by first mentioning a $\PSPACE$ lower bound for
$\AUWR$, which immediately
follows since our logic can easily encode $\PSPACE$-complete problem of 
checking emptiness of the intersection of $k$ regular languages (e.g. given as
DFA) \cite{Kozen77}.
\begin{proposition}
The satisfiability problem for $\AUWR$ is $\PSPACE$-hard.
\end{proposition}} 

We mention a natural restriction of $\AUWR$ that yields a $\PSPACE$ upper 
bound and seems to be sufficiently expressive in practice.
Recall that the \defn{dimension} of a string constraint $\varphi$ is
the maximum $\max_x$ over all string variables $x$ in $\varphi$.
[Incidentally, this notion is closely related to the notion of dimension from 
the study of context-free grammars (e.g. see \cite{EGKL11}).] 

\begin{theorem}
    For any fixed $k \in \N$, satisfiability for the class of formulas in 
    $\AUWR$ of dimension $k$ is $\PSPACE$-complete.
    \label{th:pspace}
\end{theorem}

The lower bound follows since 
our logic can easily encode in dimension one the $\PSPACE$-complete problem of 
checking emptiness of the intersection of $m$ regular languages given
as NFA \cite{Kozen77}. [In fact, when $k = 1$, the logic still subsumes
$\AUR$.] An easy corollary of the proof of Theorem \ref{theo:smp} is the
following improved bound.

\begin{theorem}
    \label{th:pspace-bmp}
    For any fixed $k \in \N$, if a formula $\varphi$ of dimension $k$ is 
    satisfiable, then it 
has a solution with each word of length at most $2^{p(|\varphi|)}$ for some 
polynomial $p(x)$.
\end{theorem}

This bound can be derived by noticing from the proof of Theorem \ref{theo:smp}
that the maximal solution size of $\varphi$ that one needs to explore is $F(x) 
:= 2^{(\Dim(\varphi) \times |\varphi|)^{O(1)}}$, which is exponential in
$|\varphi|$ if $\Dim(\varphi)$ is a fixed constant.
%
%
Note that the dimension of the scripts in 
all our examples is 2 (in fact
the dimensions of the examples in the benchmark of \cite{Stranger} are all 1,
except for one with dimension 4). 

\section{Adding Integer and Character Constraints}
\label{sec:ext}

In this section we extend the language $\AUWR$ from Section
\ref{sec:decidable_string} with {\em integer} and {\em character
constraints}, and show that the satisfiability problem remains
decidable in $\EXPSPACE$. We also show that this bound 
continues to hold in the presence of two other important features: 
the $\IndexOf$ constraints and disequalities between strings. 


Our language will use  two types of variables, $\String$ and $\Int$. The
type $\String$ consists of the string variables we considered in the
previous sections. In particular, a constraint in $\AUWR$ only uses
variables from $\String$. On the other hand, a 
variable of type $\Int$ (also called an {\em integer variable}) 
ranges over the set $\N$ of all natural numbers.
The choice of omitting negative integers is only for simplicity, 
but our results easily extend to the case when $\Int$ includes negative
integers. 
For each $T \in \{\String,\Int\}$, we use $\Var(T)$ to denote the set of
variables of type $T$. 

We start by defining integer
constraints, which allow us to express bounds for linear
combinations of lengths or number of occurences of symbols in words.  

\begin{definition}[Integer constraints] 
An \emph{atomic integer constraint} over $\Sigma$ is an expression of the form
$$a_1t_1 + \cdots + a_nt_n \, \leq \, d,$$ where $a_1, \ldots, a_n, d \in 
\mathbb{Z}$ are constant integers (represented in binary) and each $t_i$ is 
either (i) an integer variable 
$u \in \Var(\Int)$, (ii) $|x|$ for a string variable $x \in \Var(\String)$, or
(iii) $|x|_a$ for $x \in \Var(\String)$ and some constant letter 
$a \in \ialphabet$. Here,
$|x|$ (resp. $|x|_a$) denotes the length of $x$ (resp. the number of occurrences of
$a$ in $x$). 
An {\em integer constraint} over $\Sigma$ is a Boolean combination af
atomic integer constraints over $\Sigma$. 
\end{definition} 

Character constraints, on the other hand, allow us to compare symbols
from different strings. They are formally defined below. 


\begin{definition}[Character constraints]
A \emph{atomic character constraint} over $\Sigma$ is an expression of
the form $x[u] = y[v]$, where: (1) $x$ and $y$ are either 
a variable in $\Var(\String)$ or a word in $\Sigma^*$, 
and (2) $u$ and $v$ are either integer variables in $\Var(\Int)$ or
constant positive integers. Here, the interpretation of the symbol $x[u]$ is 
consistent with our notation from Section \ref{sec:prelim}, i.e., the $u$-th 
letter in $x$. A {\em character constraint} over $\Sigma$ is a Boolean 
combination af atomic character constraints over $\Sigma$.  
\end{definition} 

Next, we define the extension of the class $\AUWR$ with integer and
character constraints. 

\begin{definition}[The class $\AUWRe$] 
The class $\AUWRe$ consists of all formulas $\varphi
\wedge \theta_{\Int} \wedge \theta_{\Char}$ such that (i) $\varphi$ is a
constraint in $\AUWR$, (ii) $\theta_{\Int}$ is an integer constraint,
and (iii) $\theta_{\Char}$ is a character constraint. 
\end{definition} 

Since constraints in $\AUWRe$ are two-sorted, we have to slightly 
refine the notion of \defn{assignment}. Formally, an assignment 
for a constraint $\varphi$ in $\AUWRe$ is a mapping $\iota$ from
each variable $x \in \Var(T)$ in $\varphi$ to an object of type $T$ (i.e. either 
a string or an integer). We also assume for safety that for each term
of the form $x[u]$ in $\varphi$ it is the case that $\iota(u) \leq
|\iota(x)|$ (i.e., $\iota(u)$ is in fact a position in $\iota(x)$).
[If this assumption is not met, we can simply define that the assignment
does not satisfy the formula $\varphi$.]
As before, $\iota$ \defn{satisfies} $\varphi$ if the constraint
$\varphi$ becomes true under the substitution of each variable $x$ by
$\iota(x)$.  
We formalise this for atomic integer and character
constraints (as Boolean connectives are standard): 

\begin{enumerate}

\item $\iota$ satisfies the atomic integer constraint $\sum_{i=1}^n
  a_i t_i \leq d$ if and only if $\sum_{i=1}^n
  a_i \iota(t_i) \leq d$, where for each $1 \leq i \leq n$ we have
  that (i) $\iota(t_i) = |\iota(x)|$, if $t = |x|$
  for $x \in \Var(\String)$, and (ii) $\iota(t_i) = |\iota(x)|_a$, if $t = |x|_a$
  for $x \in \Var(\String)$ and $a \in \Sigma$. 

\item $\iota$ satisfies the atomic character constraint $x[u] = y[v]$
  if and only if $\iota(x)[\iota(u)] = \iota(y)[\iota(v)]$, where (i)
  $\iota(x) = x$ (resp., $\iota(y) = y$), if $x$ (resp., $y$) is a constant
  word over $\Sigma$, and (ii)
  $\iota(u) = u$ (resp., $\iota(v) = v$), if $u$ (resp., $v$) is a
  positive integer. 

\end{enumerate} 

The constraint $\varphi$ is \defn{satisfiable} if there exists a
satisfying assignment for it. The {\em satisfiability problem for
$\AUWRe$} is the problem of deciding if $\varphi$ is satisfiable, for a
given constraint $\varphi$ in $\AUWRe$.

\subsection{The Satisfiability Problem for $\AUWRe$} 

In this section, we will show that our $\EXPSPACE$ upper 
bound for the satisfiability of $\AUWR$ extends to $\AUWRe$.

\begin{theorem}
    The satisfiability problem for the class $\AUWRe$ is solvable in 
    $\EXPSPACE$. Furthermore, if $\AUWRe$ has a solution, then it has
    a solution of size at most $2^{2^{p(x)}}$, for some polynomial $p(x)$.
    \label{th:ext}
\end{theorem}
The rest of the section is dedicated to proving the theorem. 
Let $\AURe$ be the 
extension of $\AUR$ with integer constraints and character constraints. 
Given a formula
$\varphi \in \AUWRe$, we first transform $\varphi$ into a constraint $\varphi'$ in 
$\AURe$ by using (an extension of) the satisfiability-preserving transformation 
from Section \ref{sec:ub}. 
We will then show that the formula $\varphi'$ has a bounded model property
(cf. Lemma \ref{lm:bounded} below). More precisely, if $\varphi'$ is satisfiable,
then it has a satisfying assignment of size at most exponential in $|\varphi'|$.
This immediately provides a decision procedure for checking satisfiability of 
$\varphi$, though a naive algorithm only yields a triple-exponential procedure. 
We will show, however, that this yields a polynomial-space procedure for checking
satisfiability of $\varphi'$, and hence a single exponential space procedure for
checking satisfiability for $\varphi$.

\subsubsection{{\bf Transforming $\AUWRe$ into $\AURe$}}
\label{sec:red} 

Suppose that $\varphi = \varphi_{\String} \wedge \varphi_{\Int} \wedge 
\varphi_{\Char} \wedge \varphi_{\Reg} \in \AUWRe$, where $\varphi_{\String}$ is a 
relational constraint
in $\AUWR$, $\varphi_{\Reg}$ a regular constraint, $\varphi_{\Int}$ an integer 
constraint, and $\varphi_{\Char}$ a character constraint. 

We apply the transformations from Step 1 and 2 in 
Section \ref{sec:ub} on the
formula $\psi := \varphi_{\String} \wedge \varphi_{\Reg} \in \AUWR$, yielding an 
acyclic string/regular constraint $\psi' = \bigwedge_y \NewStrCon{y} \wedge
\NewRegCon{y}$ with no concatenation (but possibly some 
disjunctions), where each variable $x$ in $\psi$ is replaced by several 
variables $x(1),\ldots,x(\max_x)$ in $\psi'$.
Recall that Lemma \ref{lm:reduce-correctness} states that $\psi$
is satisfiable iff $\psi'$ is satisfiable. In fact, following the notation in
Step 2 of Section \ref{sec:ub}, Lemma
\ref{lm:correctness-hypothesis} shows that for each node $y$ in $\G(\psi)$ and 
every assignment $\iota: Z_y \to \Sigma^*$ of (string) variables in the
constraint $\psi_y$ (associated with the node $y$ in $\G(\psi)$),
the following two conditions are equivalent:
    
\begin{itemize}
    \item $\iota$ is a satisfying assignment for the formula $\psi_y$.
    \item There exists a satisfying assignment $\iota': Z_y' \to \Sigma^*$
        for the formula $\psi_y' := \bigwedge_{y \in Z_y} \NewStrCon{y} \wedge 
        \NewRegCon{y}$ such that for each $x \preceq y$: 
        \begin{equation*}
            \iota(x) = \iota'(x(1))\circ \cdots\circ \iota'(x({\max}_x)). \tag{*}
        \end{equation*}
    \end{itemize}

Handling the integer constraint is now easy. Because of (*),
we simply replace 
each occurrence of $|y|$ (resp. $|y|_a$, where $a \in \ialphabet$) in
$\varphi_{\Int}$ by
$\sum_{i=1}^{{\max}_y} |y(i)|$ (resp. $\sum_{i=1}^{{\max}_y} |y(i)|_a$). 
Let $\psi_{\Int}'$ be the resulting formula.
[Note that $y$ is \emph{not} a variable in $\psi'$.]
This transformation preserves satisfiability, even in the above stronger sense.

Let us now show how to deal with the character constraint $\varphi_{\Char}$.
Without loss of generality, we may assume that the term $t$ which occurs on 
left/right hand side of atomic character constraint is of the form $x[u]$ (for an 
integer variable $u$), which denotes the $u$-th character in $x$. [If $t$ is a string
variable $x$, we can replace $x$ by $x[u]$, where $u$ is a fresh $\Int$ variable, 
and add the integer constraint
$|x| = 1 \wedge u = 1$. Similarly, if $t$ is of the form $x[c]$, where $c$
is an integer constant, we could simply replace this by $x[u]$, for a fresh
$\Int$ variable $u$, and add the integer constraint $u = c$.] Now the $u$-th
character $x[u]$ in $x$ must fall within
precisely one of the word segments $x(1), \ldots, x(\max_x)$. Therefore, 
we simply make a nondeterministic guess on which segment $x(i)$ the position $x[u]$ 
belongs to, and replace
every occurrence of $x[u]$ by $x(i)[u']$, where $u'$ is a fresh $\Int$ variable,
and add an integer constraint of the form $u = u' + \sum_{j=1}^{i-1} |x(j)|$.
Observe that the constraint $\theta$ that we generate from $\varphi_{\Char}$ 
involves both integer constraint and character constraints. 
Then, the formula $\varphi$ is satisfiable iff, for some formula 
$\theta$ obtained from the aforementioned nondeterministic construction,
the formula $\psi' \wedge \psi_{\Int}' \wedge \theta$ is satisfiable.

Note, however, that $\psi'$ still has some disjunctions and so strictly
speaking it is not a formula in $\AUR$.  So, to complete our transformation of
$\varphi$ into $\AURe$, we use Step 3 from Section \ref{sec:ub}
on $\psi'$ to make further nondeterministic guesses to eliminate the disjunctions. 
Let us call a possible resulting formula $\psi''$. Therefore,
$\varphi$ is satisfiable iff, for some $\psi''$ and $\theta$, the formula 
$\psi'' \wedge \psi_{\Int}' \wedge \theta$ is satisfiable.

\OMIT{
Lemma \ref{lm:correctness-hypothesis} implies
more. 
Note that this 
transformation is \emph{nondeterministic}, so
the only guarantee that is provided is that 
}

\subsubsection{{\bf Bounded Model Property for $\AURe$}}
We will prove a bounded model property for $\AURe$.

\begin{lemma}[Bounded Model]
    Given a formula $\varphi$ in $\AURe$, if it is satisfiable, then there exists a
    satisfying assignment of whose strings are of length at most exponential 
    in $|\varphi|$ and whose integers are of size (in binary) at most polynomial in
    $|\varphi|$.
    \label{lm:bounded}
\end{lemma}

Before proving this lemma, we will first show how this can be used to
obtain Theorem \ref{th:ext}.

\subsubsection{Lemma \ref{lm:bounded} Implies Theorem \ref{th:ext}}

As mentioned in Section \ref{sec:red}, the problem of checking whether
an $\AUWRe$ formula $\varphi$ is satisfiable can be reduced in
nondeterministic exponential time to checking whether an $\AURe$
formula $\varphi'$ is satisfiable. Next we construct an 
algorithm that solves this problem in $\PSPACE$ in the size of
$\varphi'$, and, therefore, in $\EXPSPACE$ in the size of $\varphi$
(since $\varphi'$ might be exponentially bigger than $\varphi$).
Theorem \ref{th:ext} then follows from the fact that nondeterministic
exponential time is contained in $\EXPSPACE$ and $\EXPSPACE$
computable functions are closed under composition. 

In order to do this, the first step is to
restate a stronger form of Proposition \ref{prop:BFL-PSPACE}. 
To this end,
we first recall the standard generalisation of the notion of transducers to
allow an arbitrary number of tracks (e.g. see \cite{CC06}).
An \defn{$m$-track (rational) transducer} over the alphabet $\ialphabet$ is a
tuple $\Aut = (\Gamma,\controls,\transrel,q_0,\finals)$, where
$\Gamma := \ialphabet_\epsilon^m$ and $\ialphabet_\epsilon := \ialphabet
\cup \{\epsilon\}$, such that $\Aut$ is syntactically an NFA over $\Gamma$.
In addition, we define the $m$-ary relation $R \subseteq (\ialphabet^*)^m$ that
$\Aut$ \emph{recognises} to consist of all tuples $\bar w$ for which there is an accepting 
run 
\[
    \pi := q_0 \tran{\sigma_1} q_1 \tran{\sigma_2} \cdots \tran{\sigma_n} q_n
\]
of $\Aut$ (treated as an NFA) such that $\bar w = \sigma_1 \circ \sigma_2 
\circ \cdots \circ \sigma_n$, where the string concatenation operator $\circ$ 
is extended to tuples over words component-wise (i.e. $(v_1,\ldots,v_k) 
\circ (w_1,\ldots,w_k) = (v_1w_1,\ldots,v_kw_k)$). An $m$-ary relation is said
to be \defn{rational} if it is recognised by an $m$-ary transducer. To avoid 
notational clutter, we shall confuse an $m$-ary transducer and the $m$-ary 
relation that it recognises. In the following proposition, the set of solutions
(i.e. satisfying assignments) to a formula $\varphi$ in $\AUWRe$ --- i.e. mappings
from variables $x$ in $\varphi$ to strings, integers, or characters depending
on the type of $x$ --- is interpreted as a relation (i.e. a set of tuples) by
fixing any ordering to the variables occuring in $\varphi$.

\begin{proposition}[\cite{BFL13}]
    There exists an exponential-time algorithm for computing an $m$-track 
    transducer $\Aut = (\Gamma,\controls,\transrel,q_0,\finals)$ for the set of solutions 
    of an 
    input formula $\varphi \in \AUR$ with variables $x_1,\ldots,x_m$, where each 
    state in $\controls$ is of size polynomial in $|\varphi|$. Furthermore, there 
    exists a polynomial space algorithm for:
    \begin{enumerate}
    \item Computing $q_0$. 
    \item Checking whether a string is a state of $\controls$.
    \item Checking whether a state $q$ belongs to $\finals$.  
    \item Checking whether $(q,\bar a,q') \in \delta$, for some given states $q$ and 
        $q'$ and a symbol $\bar a \in \Gamma$. 
    \end{enumerate}
    \label{prop:BFL-transducer}
\end{proposition}
This proposition is a stronger version of Proposition \ref{prop:BFL-PSPACE}, which 
follows from the proof of Theorem 6.7 in \cite{BFL13}.

Obtaining a $\PSPACE$ algorithm for checking satisfiability of a given formula
$\varphi'$ in $\AURe$ is now almost immediate. Our nondeterministic algorithm guesses
an assignment to each character variable in $\varphi$, whose size is linear in
$|\varphi|$. By virtue of Lemma \ref{lm:bounded}, our algorithm needs to guess
an assignment to each integer variable in $\varphi$ of size polynomial in
$|\varphi|$ (i.e. numbers represented in binary). In effect, if $|x|$ (resp.
$|x|_a$) appears in the integer constraint of $\varphi$, our algorithm
also guesses the length of (resp. number of occurrences of $a$ in) the string 
variable $x$. Our algorithm now checks that
the integer constraints and character constraints in $\varphi$ are satisfied.
Next our algorithm guesses assignments to the string variables in $\varphi$.
However, since they are of exponential size in the worst case, our algorithm
will have to construct the assignment on the fly. By using
Proposition \ref{prop:BFL-transducer}, we will have to \emph{simultaneously} 
construct the string assignments to all string variables in $\varphi$.
To this end, for each string variable 
$x$, we keep track of $|\ialphabet|+1$ extra integer counters (counting
in binary) $l(x)$ and $l_a(x)$, respectively, for each $a \in\ialphabet$. The counter $l(x)$
(resp. $l_a$) keeps track of the length of (resp. number of occurrences of $a$ in)
the partially constructed string assignment for $x$. Putting this all together,
if $\Aut = (\Gamma,\controls,\transrel,q_0,\finals)$ is the $m$-track transducer for the 
set of solutions for $\varphi$ (as in Proposition \ref{prop:BFL-transducer}),
our algorithm first computes $q_0$ and let $q = q_0$. It then repeats the following 
step until (i) $q \in F$, and (ii)  
$l(x) = |x|$ and $l_a(x) = |x|_a$, for each string variable $x$ and letter $a \in 
\ialphabet$: 
\begin{enumerate}
\item 
Guess a state $q'$ of $\controls$ and
a symbol $\bar a \in \Gamma = \ialphabet_\epsilon^m$ of $\Aut$, and check
that $(q,\bar a,q') \in \transrel$ in polynomial space. 
\item If 
$\bar a = (a_1,\ldots,a_m)$, then set $l(x_i) := l(x_i) + 1$ and $l_a(x_i) := 
l_a(x_i) + 1$ for each $i \in \{1,\dots,m\}$ and $a \in \ialphabet$
satisfying $a_i = a$.  
\item Set $q := q'$
\end{enumerate} 

This is a nondeterministic algorithm that uses polynomial space. 
Nondeterministic algorithms can be determinised at the cost of quadratically
extra space \cite{Savitch70}. As for the bounded 
model property part of Theorem \ref{th:ext}, it can be derived in precisely
the same way as Theorem \ref{theo:smp}.

\OMIT{
This lemma immediately gives decidability of $\AUWRe$: simply enumerate all
satisfying assignments within the bound given in Lemma \ref{lm:bounded}, and
check whether one of them is a satisfying assignment of $\varphi$.
We first show how to use Lemma \ref{lm:bounded} to obtain an $\EXPSPACE$
algorithm for checking satisfiability of formulas in $\AUWRe$. Later on in
this section, we will then show how to prove Lemma \ref{lm:bounded}.
}

\subsubsection{Proving Lemma \ref{lm:bounded}}

Our proof idea goes as follows. Given a satisfiable formula $\varphi \in \AURe$ with
$m$ variables,
we suppose that $\iota$ is a satisfying assignment
of $\varphi$. It assigns each character variable $x[u]$ to a certain character
$\iota(x[u]) \in \ialphabet$. So, we will only have to ensure the existence of
a satisfying assignment that assigns each string (resp. integer) variable to a small 
enough string (resp. integer). To this end, Proposition \ref{prop:BFL-transducer} 
gives an 
$m$-track transducer $\Aut = (\Gamma,\controls,\transrel,q_0,\finals)$ that recognises
the set of solutions of $\varphi$. The number of states in $\Aut$ is exponential
in $|\varphi|$.
Next we erase the input tape of $\Aut$, while equipping it with nonnegative 
integer-valued counters that \emph{cannot
be decremented}. The resulting machine 
is a \emph{monotonic (Minsky's) counter machine}, whose set of final configurations
captures the set of solutions for $\varphi$.
\OMIT{

%
}
Monotonic counter machines are restrictions of reversal-bounded counter machine
\cite{Iba78} (i.e. counter machines whose counters 
can switch between 
non-incrementing and non-decrementing modes only for a fixed $r \in \N$ number of
times). [In the case of monotonic counter machines, we have $r = 0$].
Such a computation model is \emph{not} Turing-complete. In fact, their sets of
reachable configurations are effectively semi-linear, as was first shown in 
\cite{Iba78}. We will use a recent result from \cite{KT10,anthony-thesis} to 
analyse the size of the smallest reachable configuration, which will give us the
bounded model property of $\AURe$.

We now formalise the notion of monotonic counter machine with $k \in \N$ counters. 
A \defn{monotonic counter machine} 
is a tuple $\Aut = (\controls,\transrel,q_0,\finals,P)$, where: 
(1) $\controls$ is a finite set of states, (2) $q_0 \in \controls$ is an initial
state, (3) $\finals \subseteq \controls$ is a set of final states, (4)
$P = \{p_1,\ldots,p_k\}$ is a set of $k$ counters, and (5) $\transrel
\subseteq (\controls \times \Cons_{P}) \times 
(\controls \times \{0,1\}^k)$ is the transition relation, where 
$\Cons_P$ is the set of counter tests of the form $\bigwedge_{i=1}^k
p_i \sim_i 0$ such that $\sim_i \ \in \{=,>\}$ for each $1 \leq i \leq k$. 
A vector $\bar v \in \{0,1\}^k$ such that $(q,\psi,q',\bar v) \in
\transrel$, for $q,q' \in Q$ and $\psi \in \Cons_P$, 
is called an \emph{update vector}. In the sequel we shall also denote
this vector by its characteristic set, i.e., the one which consists 
of all counters $p_i \in P$ such
that $\bar v[i] = 1$.

A \defn{configuration} of $\Aut$ is a pair $(q,\bar v)$, where $q \in \controls$
and $\bar v \in \N^k$. A \defn{run} $\pi$ of $\Aut$ is a sequence of the form
\[
    (q_0,\bar v_0), \, (q_1,\bar v_1), \, \dots, \,(q_n,\bar v_n)
\]
such that: 
\begin{itemize}
    \item $q_i \in \controls$ for each $1 \leq i \leq n$,
    \item $\bar v_0 = (0,\ldots,0)$ (i.e., counters are initially
empty), 
    \item for each $1 \leq i \leq n$ there exists a transition $(q_{i-1},\psi(\bar x),
        q_{i},\bar c) \in \delta$ such that $\psi(\bar v_{i-1})$ is true, and 
        $\bar v_i = \bar v_{i-1} + \bar c$.
\end{itemize}
The configuration $(q_n,\bar v_n)$ is said to be \defn{accepting} if 
$q_n \in \finals$. The set of accepting configurations of $\Aut$ is denoted by
$\Lang(\Aut)$. 

From our transducer $\Aut = (\Gamma,\controls,\transrel,q_0,\finals)$, we construct
our monotonic counter machine $\AutB = (\controls,\transrel',q_0,\finals,P)$,
where $P$ consists of the following counters:
\begin{itemize}
    \item $c_{|x|_a}$ for each string variable $x$ in $\varphi$ and each 
        letter $a \in \ialphabet$;
    \item $c_u$ for each integer variable $u$ in $\varphi$; 
    \item $y_{x[u]}$ and $z_{x[u]}$ for each $x[u]$ occuring in a character 
        constraint of $\varphi$.
\end{itemize}
The counter $c_{|x|_a}$ records the number of $a$'s seen so far in
in the transducer's track corresponding to the variable $x$.
The counter $c_u$ records the guessed value for the variable $u$.
To avoid notational clutter, we shall confuse $c_{|x|_a}$ (resp. $c_u$) with 
$|x|_a$ (resp. $u$).
The counter $y_{x[u]}$ records a guess for the position of $u$ in $x$, which has
to be recorded separately due to the different tracks in $\Aut$.
The value of $z_{x[u]}$ is a boolean variable (i.e., either 0 or 1) that 
acts as a flag whether the guess for the variable $y_{x[u]}$ is complete.

We now specify the transitions in $\transrel'$. In doing so, we will ensure
that once a variable of the form $z_{x[u]}$ is set to 1, the value of $y_{x[u]}$
can no longer be incremented. Let $W$ be the set of all character variables $x[u]$
in $\varphi$. For each $Y \subseteq W$, let $\psi_Y$ denote the formula of the 
form
$\bigwedge_{x[u] \in Y} z_{x[u]} = 0 \wedge \bigwedge_{x[u] \in W\setminus 
Y} z_{x[u]} > 0$. Then: 
\begin{enumerate}
    \item If $(q,\bar a,q') \in \transrel$, where $\bar a = (a_1,\ldots,a_m)$,
        then for each subset $Y\subseteq W$ we add to $\delta'$
\emph{each} transition of the form
        \[
            (q,\psi_Y,q',Z), 
        \]
        where $Z$ consists of: (1) each $|x|_{a_i}$ with $a_i \neq \epsilon$,
        (2) $y_{x[u]}$ for each $x[u] \in Y$, and (3) 
        if $z_{x_i[u]} = 0$ and $a_i = \iota(x_i[u])$, the set $Z$ may
        (nondeterministically) contain $z_{x_i[u]}$.
    \item $(q,\top,q,\{u\})$ for each integer variable $u$ in
$\varphi$. 
\end{enumerate}
In other words, the first transition above simply simulates a transition of
$\Aut$, while the second transition nondeterministically increments the integer
counter $u$. 

Recall that $\varphi$ is our initial formula in $\AURe$ and $\iota$ is
a satisfying assignment for it. 
Now, let $\varphi_\Int$ be the conjunct of
$\varphi$ containing the integer constraint. We use $\varphi_\Int'$ to denote
a conjunction of (i) the constraint $\varphi_\Int$ but substituting every 
occurrence of $|x|$ by $\sum_{a\in\ialphabet} |x|_a$, (ii) a conjunction of
constraints of the form $u = y_{x[u]}$ for each $x[u] \in W$ 
\anthonychanged{(i.e. all positions $y_{x[u]}$ equal $u$)}, and (iii) a 
conjunction of constraints of the form $z_{x[u]} = 1$ for each $x[u] \in W$
\anthonychanged{(i.e. all $y_{x[u]}$ have been completely guessed)}.
The following lemma is immediate from our construction and Proposition 
\ref{prop:BFL-transducer}.

\begin{lemma}
    Given a configuration $(q,\bar v)$ of $\AutB$, the following are
    equivalent:
    \begin{itemize}
        \item $(q,\bar v) \in \Lang(\AutB)$ and $\bar v$ satisfies
$\varphi_\Int'$. 
        \item There exists a satisfying assignment $\iota'$ of $\varphi$ 
            whose integer values agree with $\bar v$ and whose character
            values agree with $\iota$.
    \end{itemize}
    \label{lm:correctness-counter}
\end{lemma}

We now use the following proposition, which is a result of \cite{KT10} (see the 
proof of \cite[Proposition 7.5.5]{anthony-thesis}):
\begin{proposition}
    Given a monotonic $k$-counter machine $\mcl{A}$ with $n$ states, the
    set of reachable configurations can be represented as a disjunction of
    existential Presburger formulas, each of size polynomial in 
    $k+\log(n)$ and at most $O(k)$ variables.
\end{proposition}
Indeed, in the above proposition the number of disjuncts is polynomial in
$n$, but this is not important for our purpose. It is now easy to obtain a
satisfying assignment for $\varphi$ given our original satisfying assignment
$\iota$. By the above proposition, the set of reachable configurations of the 
monotonic counter machine 
$\mcl{B}$ that we constructed is a disjunction of existential Presburger formulas 
each of size 
polynomial in $k+\log(n)$ and with $O(k)$ variables, where $k = O(|\varphi|)$ and 
$n = 2^{|\varphi|^{O(1)}}$. By Lemma \ref{lm:correctness-counter} and our
assumption that $\varphi$ is satisfiable with assignment $\iota$, it follows that 
one of these disjuncts $\psi$ is satisfiable.
Scarpellini \cite[Theorem 6(a)]{Sca84} proved that
a satisfiable existential Presburger formula $\theta$ with $c$ variables has
solutions where each variable is assigned a number that is at most 
$2^{(|\theta|+c)^{O(1)}}$ (and thus can be represented with at most $(|\theta|+c)^{O(1)}$
bits). Applying Scarpellini's result on $\psi$ now gives us a satisfying
assignment of $\varphi$ which assigns numbers of polynomially many bits to
integer variables of $\varphi$ and lengths of string variables. This
completes the proof of Lemma \ref{lm:bounded}.

\subsection{Extensions with Disequalities and $\IndexOf$}

Finally, we show that two important features can be added to the
language while retaining decidability in $\EXPSPACE$: Disequalities
between strings and $\IndexOf$ constraints. 

\paragraph*{{\bf Disequalities:}} Assume that 
constraints in $\AUWRe$ are now extended 
with disequalities of the form $x
\neq y$, for $x,y \in \Var(\String)$, which state that $x$ and $y$
are interpreted as different strings. The disequality relation is
regular, and thus can be expressed as a transducer $y =
R(x)$. The problem is that the addition of this
transducer may yield a constraint that is no longer uniquely
definitional. To solve this, we use 
integer and character constraints; in fact, the disequality $x \neq y$
is equivalent to 
$(|x| \neq |y|) \vee (x[u] \neq y[u])$, for a fresh variable $u \in \Var(\Int)$. More
formally, if $\varphi$ is a constraint in $\AUWRe$ and $x \neq y$ is a
disequality between string variables, then $\varphi \wedge (x \neq y)$
is satisfiable if and only if either the $\AUWRe$ constraint $\varphi \wedge
(|x| \neq |y|)$ or the $\AUWRe$ constraint $\varphi \wedge 
(x[u] \neq y[u])$ is satisfiable. Checking if any of these constraints
is satisfiable can be solved in $\EXPSPACE$ from Theorem \ref{th:ext}.
Clearly, adding more disequalities does not increase the computational
cost if we use a nondeterministic algorithm that chooses to check
either $|x| \neq |y|$ or $x[u] \neq y[u]$ for each disequality $x \neq
y$. 

\OMIT{
A \defn{constraint in our language} is simply a conjunction of a string constraint,
an integer constraint, a character constraint, and a regular constraint.
}

\paragraph*{{\bf Expressing the $\IndexOf$ method:}}
One reason we introduced the character constraints is, besides the use of the 
JavaScript string 
method \texttt{charAt} (which is used rather frequently in JavaScript according
to the benchmark \cite{Berkeley-JavaScript}), they can also be used to
define $\IndexOf(w,x)$ for any word $w \in \ialphabet^*$, which is
the most standard usage of $\IndexOf$ method in practice. We consider both
the \emph{first-occurrence} semantics (i.e., for an integer variable $u$,
the constraint $u = \IndexOf(w,x)$ says that $u$ is the first position in $x$ where 
$w$ occurs), or the \defn{anywhere} semantics (i.e., $u = \IndexOf(w,x)$ says that
$u$ is any position in $x$ where $w$ occurs).

Formally, an $\IndexOf$ constraint is a conjunction of expressions of
the form $u = \IndexOf(w,x)$, where $w \in \Sigma^*$, $x$ is a 
string variable/constant, and $u$ is either an
integer variable or a positive integer. The satisfaction of an
expression of this form 
(under any of the two semantics) with
respect to an assignment of the variables is the expected one (following the intuition
given in the previous paragraph). The next proposition states
that $\IndexOf$ constraints 
do not increase the ``expressiveness'' of $\AUWRe$.  

\begin{proposition}
Let $\varphi$ be the conjunction of a constraint in $\AUWRe$ and the
$\IndexOf$ constraint $u = \IndexOf(w,x)$. The satisfiability of
$\varphi$ can be checked in $\EXPSPACE$.  
    \label{prop:indexof}
\end{proposition}

\begin{proof}
    Let $w = a_1\ldots a_p \in \ialphabet^*$. For the anywhere semantics, every 
    occurrence of $u = \IndexOf(w,x)$ in $\varphi$ is replaced by the formula
    \begin{multline*}
        x[u_1] = a_1 \wedge \cdots \wedge x[u_p] = a_p 
    \\ 
        u = u_1 \wedge u_2 = u_1 + 1 \wedge \cdots \wedge u_p = u_{p-1} + 1,
    \end{multline*} 
where $u_1,\dots,u_p$ are fresh integer variables. The resulting
formula is an
$\AUWRe$ constraint whose satisfiability can be checked in $\EXPSPACE$
from Theorem \ref{th:ext}.  
    
The first-occurrence semantics could be handled by replacing $u =
\IndexOf(w,x)$ with a  
    relational constraint $x = x_1 x_2 x_3$
    and regular constraints stating that $w$ does not occur in $x_1$ and $w = x_2$. 
This would allow us
    to use the anywhere semantics to express $u = \IndexOf(w,x_2)$, which can be
    done as before. The problem with this approach is that the
introduction of the word equation $x = x_1 x_2 x_3$ may yield a
constraint that is no longer uniquely definitional.  
To avoid this, we make a nondeterministic
guess (as we did for formula $\varphi_{\Char}$ in Section
\ref{sec:red}) 
as to how $x_1$ overlaps with
$x(1),\ldots,x(\max_x)$, e.g., it might overlap with $x(1)x(2)x(3)$. We will
then simply assert that $x(1)x(2)x(3) \in L$, where $L$ is the regular
language of words $v$ that contains no $w$ as a (contiguous) subword.
We then have to apply the splitting technique for the regular
constraint just as in Step of 
Section \ref{sec:decidable_string} to express $x(1)x(2)x(3) \in L$ as
$\bigwedge_{i=1}^3 x(i) \in L_i$. This can all be done by nondeterministic guesses
while incurring only a polynomial blowup.
\end{proof}

\section{Related Work and Future Work}
\label{sec:rw} 

In this section, we mention a few related works and discuss their connections
with our work in more detail. Roughly
speaking, they can be classified into three categories: (1) decidability 
results, (2) heuristics and string solver implementations, (3) benchmarking 
examples. We shall also mention a few possible research avenues in passing.

\paragraph*{{\bf Decidability results:}}

In \S Introduction, we have mentioned the results of Makanin's and 
Plandowski's 
\cite{Makanin,Plandowski,Plandowski99,Plandowski06} on the 
decidability and complexity of satisfiability for word equations (a conjunction 
of equations of the form $v = w$, where $v$ and $w$ are a concatenation of 
string constants and variables). We should also remark that the decidability
(with the same $\PSPACE$ complexity)
extends to quantifier-free first-order theory of strings with concatenations
and regular constraints \cite{buchi}. Since extending word equations with
finite-state transducers yields undecidability (see Section \ref{sec:lang}),
the straight-line fragment $\AUWR$ of our core string constraint language is 
incomparable to word equations with regular constraints (neither subsumes the 
other). The fragment $\AUWR$ is, in a sense,
more complex since its computational complexity is $\EXPSPACE$, though 
for constraints of small dimensions the complexity reduces to $\PSPACE$
(cf. Theorem \ref{th:pspace}). 
In addition, it is still a long-standing open problem whether word equations 
with 
length constraints is decidable, though it is known that letter-counting (i.e. 
counting the number of occurrences of 0s and the number of occurrences of 1s 
separately) yields undecidability \cite{buchi}. On the other hand, the extension
$\AUWRe$ of our straight-line fragment $\AUWR$ is decidable (with the same 
complexity) and yet admits general letter-counting. 

In our decidability proof of Theorem \ref{th:expspace}, we have also used the 
result of Barcelo \emph{et al.} \cite{BFL13} (see Proposition 
\ref{prop:BFL-PSPACE}) that acyclic conjunctions of rational relation 
constraints (with regular constraints) is decidable in $\PSPACE$. 
Their logic $\sf AC$, however, supports neither string concatenations nor 
letter-counting constraints. In fact, it is easy to show by a standard pumping 
lemma argument that the constraint $x = y\cdot y$ cannot be expressed
in $\sf AC$. For this reason, our logics $\AUWR$ and $\AUWRe$ are not
subsumed in $\sf AC$. 

Abdulla \emph{et al.} \cite{Abdulla14} studied acyclic constraints over 
systems of word equations with a length predicate (without transducers)
and disequality constraints, for which they showed decidability. Our
decidable logics are incomparable to their logic.
On the one hand, $\AUWRe$ supports finite-state transducers, letter-counting,
and \texttt{IndexOf} constraints, which are not supported by their logic.
Our logic also supports unrestricted disequality constraints, 
whereas their logic supports only restricted (acyclic) disequalities.
On the other hand, the string logic of \cite{Abdulla14} supports 
string equations of the form $x\cdot y = z\cdot z'$ (i.e. both sides of the
equations contain different variables), which they showed could be reduced
to a boolean combination of regular constraints. Using this reduction,
we can incorporate this feature into $\AUWRe$, yielding a more expressive 
decidable string logic.
%

\OMIT{
Satisfiability for word equations 
were proven to be decidable by Makanin \cite{Makanin} in 1977. The 
complexity of Makanin's algorithm is super-exponential. There were a lot 
of works on improving the complexity of Makanin's algorithm \cite{??} 
\anthony{Pablo?}
culminating in Plandowski's polynomial space (i.e. single-exponential time)
algorithm \cite{Plandowski99,Plandowski} in 1999. Plandowski's 
algorithms also ...
}

\OMIT{
Mention symbolic transducers, symbolic finite-state transducers, ...
}

\paragraph*{{\bf Heuristics and string solver implementation:}}
In \S Introduction, we have mentioned the large amount of works in the past
seven years or so towards developing practical string solvers (e.g. 
\cite{Stranger,fmsd14,Berkeley-JavaScript,S3,BTV09,Abdulla14,cvc4,HW12,RVG12,Yu09,Balzarotti08,patching,symbolic-transducer,ganesh-word,SMV12,Was07,Was08a,Was08b,FL10,FPBL13,BEK,HAMPI,Z3-str,yu2011,CMS03,GSD04,Min05,DV13}). 
We are not aware of existing string solvers that support both 
concatenations and finite-state transductions. However, string solvers that 
support concatenations and the \texttt{replace-all} operator (i.e. a subset of
finite-state transductions) are available, e.g., \cite{S3,fmsd14,BTV09}. 

Since the focus of our work is on the fundamental issue of decidability,
we consider our work to be complementary to these works. In fact,
our results may be construed as providing some \emph{completeness guarantee} for
existing string solvers. 
Practical string solvers do not implement Makanin's or Plandowski's algorithms,
but instead rely on certain heuristics (e.g.  bounding the maximum length $k$ 
of solutions \cite{Berkeley-JavaScript,BTV09,HAMPI}). For this reason, none of 
the above solvers have a completeness guarantee for the entire class of word 
equations. 
However, when the 
input string constraint $\varphi$ falls within the logic $\AUWRe$, the bounded
model properties of $\AUWR$ and $\AUWRe$ (e.g. Theorem \ref{theo:smp} and 
Theorem \ref{th:ext}) 
imply that string solvers need only look for solutions of size at most 
$2^{2^{p(|\varphi|)}}$ for some polynomial $p(x)$. Double exponential size 
is of course only an extremely crude estimate, so in practice one could devise
an algorithm for computing a better estimate $f(\varphi) \leq
2^{2^{p(|\varphi|)}}$ by looking at the structure of the formula $\varphi$.
A rough estimate of $f(\varphi)$ could, for example, be obtained by first
computing the dimension of $\varphi$ (which could be computed quickly);
as we have seen in Theorem \ref{th:pspace-bmp}, when the dimension is 
\emph{small} the double exponential bound actually reduces to exponential size.
\OMIT{
We leave the construction of a better algorithm for estimating $f(\varphi)$ for 
future work.}
\begin{future}
Give a better algorithm for computing a better estimate $f(\varphi)$ of the 
maximum size of the solutions for straight-line formulas $\varphi$ that need to 
be explored.
\end{future}

\OMIT{
All of these solvers circumvent the undecidability of the logic by different
heuristics. 
The solver of \cite{BTV09} first finds a satisfying assignment 
$\nu$ to the length abstraction of the given string constraint $\varphi$ using
a solver for linear arithmetic, and then applies bounded-length string solver 
to find a string concretisation of the integer variables in the abstraction
(if no solution is found here, backtrack). The solver of \cite{S3} finds 
a satisfying solution of the constraint $\varphi$ by folding and unfolding the 
recursive definitions of each operator. Termination guarantee can be 
achieved for solutions of fixed length. Finally, 
Discuss \cite{S3,fmsd14}.
}


Veanes \emph{et al.} (e.g. see \cite{HV11,symbolic-transducer,DV13,BEK}) have 
observed 
that, in practice, the number of transitions of finite-state transducers for 
encoding web sanitisation functions could become large fairly quickly
(due to large alphabet size, e.g., \texttt{utf-8}). For this reason,
they introduce extensions of finite-state transducers that allow succinct
representations 
by allowing transitions to
take an arbitrary formula in a decidable logical theory, while taking advantage
of state-of-the-art SMT solvers for the theory. As a simple example, 
consider the finite-state transducer that converts a sequence of digits
(over the alphabet 
$\ialphabet = \{\text{\texttt{0}},\ldots,\text{\texttt{9}}\}$) 
to its HTML character numbers (over the alphabet $\ialphabet' = 
\{\text{\texttt{\&}},\text{\texttt{;}},\text{\texttt{\#}},\text{\texttt{0}},
\text{\texttt{9}}\}$). The general formula for this is that
a digit \text{\texttt{i}} is converted to \texttt{\&\#(48+i);}. Using the
standard finite-state transducers, we would require about $2+10 \times 3 = 32$
states, whereas representing using symbolic finite-state transducers only
$\sim 4$ states and $\sim 4$ transitions are required. There are real-world 
examples where 
this compression would be enormous (e.g. see the encoding of \texttt{HTMLdecode}
in \cite{symbolic-transducer} as a symbolic transducer). For this reason, 
for future work, it would make sense to consider an extension of our work that 
uses symbolic finite-state transducers (or extensions thereof) both from
practical and theoretical viewpoints.
\begin{future}
    Study the extension of $\AUWR$ and $\AUWRe$ with symbolic (finite-state)
transducers. 
\end{future}

In order to be able to apply string solvers to analyse injection and XSS
vulnerabilities, it is paramount to develop realistic browser models, which
would model implicit browser transductions. Preliminary works in this
direction are available (e.g. see \cite{web-model,Scriptgard}). 
We believe that an interesting (but perhaps extremely challenging) line of 
future work is to develop a formal and precise browser transductions and their 
transductions (e.g. for a particular version of Firefox) as finite-state 
transducers (or extensions thereof). Although some such transducers
are already available (e.g. see 
\cite{BEK,DV13,symbolic-transducer,BEK-website}), much work remains to be
done to develop full-fledged browser models that can capture all the
subtlety of browser behaviors (e.g. those that can be found in
\cite{mXSS,html5sec}).


%
%


\paragraph*{{\bf Benchmarking examples:}}
In this paper, we have provided four examples of analysis of mutation-based XSS
vulnerabilities 
that can be expressed in $\AUWR$. A few other interesting XSS vulnerability
examples from \cite{mXSS,Stock14} and \cite{html5sec} can actually be expressed 
in our 
logic, though the vulnerabilities only exist in older browsers (e.g. IE8).
We also note that the five benchmarking
examples of PHP programs from \cite{Stranger} that exhibit SQL injection and
XSS vulnerabilities (cf.  \url{http://www.cs.ucsb.edu/~vlab/stranger/}) can also
be expressed in $\AUWR$. As we have argued in \S Introduction, the benchmarking
examples from Kaluza \cite{Berkeley-JavaScript} are in \emph{solved forms},
and therefore expressible in $\AUWRe$. To the best of our knowledge, the
benchmarking examples from \cite{Stranger} and \cite{Berkeley-JavaScript} do
not contain mutation XSS test cases.

\acks
We thank Leonid Libkin and anonymous referees for their helpful feedback. We 
also thank Lukas Holik and Joxan Jaffar for the fruitful
discussion. Lin was supported by Yale-NUS College through the MoE Tier-1 grants
R--607--265--056--121 and IG15--LR001. Barcel\'o is funded by the Millenium 
Nucleus Center for Semantic Web Research under grant NC120004.

\bibliographystyle{abbrvnat}

\bibliography{references}

\shortlong{}{
\section*{Appendix} 

\subsection*{Proof of Proposition \ref{prop:undec-prog}:} 

We use a reduction from PCP. Recall that an 
input to PCP are two equally long lists $a_1,a_2,\dots,a_n$
and $b_1,b_2,\dots,b_n$ of words over alphabet $\Sigma$. 
We want to decide whether there exists a {\em solution} for this input, i.e., 
a nonempty sequence of indices $i_1,i_2,\dots,i_k$ such that $1 \leq i_j
\leq n$ ($1 \leq j \leq k$) and $a_{i_1} a_{i_2} \cdots a_{i_k} =
b_{i_1} b_{i_2} \cdots b_{i_k}$.

Before explaining how our reduction works, we introduce some
notation. 
Assume without loss of generality that $\Sigma$ is disjoint from
$\mathbb{N}$. With every input $a_1,a_2,\dots,a_n$ and
$b_1,b_2,\dots,b_n$ for PCP over alphabet $\Sigma$, we define 
an alphabet $\Sigma(n) := \Sigma \cup
\{1,2,\dots,n\}$. 
We also define 
regular expressions
$L_i$ over $\Sigma(n)$, for $1 \leq i \leq 6$, 
 as follows: 
(1) 
 $L_1 := (\bigcup_{1 \leq i \leq n} a_i
\cdot i)^+$, (2) $L_2 = L_3 = L_5 = \Sigma(n)^+$, (3) $L_4 := (\bigcup_{1 \leq j \leq n} b_j
\cdot j)^+$, and (4)  
$L_6 = \epsilon$. 

Further, define $A = C = \{1,\dots,n\}$ and $B = D = \Sigma$.
Thus, if $w$ is a word over $\Sigma(n)$ it is the case that 
$\replace{w}{\epsilon/A} =$ $\replace{w}{\epsilon/C}$ corresponds to the
word obtained from $w$ by deleting every symbol from
$\{1,\dots,n\}$. On the other hand,  
$\replace{w}{\epsilon/B} = \replace{w}{\epsilon/D}$ corresponds to the
string obtained from $w$ by deleting every symbol from $\Sigma$.  

We now explain our reduction. Consider an input to PCP given by lists 
$a_1,a_2,\dots,a_n$ and
$b_1,b_2,\dots,b_n$ of words over $\Sigma$. We then construct a string
constraint
$\varphi$ as follows:  
\begin{eqnarray} \label{eq:pcp} 
    x_2 & = & \replace{x_1}{\epsilon/A} \wedge \nonumber \\
    x_3 & = & \replace{x_1}{\epsilon/B} \wedge \nonumber \\
    x_2 & = & \replace{x_4}{\epsilon/C} \wedge \nonumber \\
    x_3 & = & \replace{x_5}{\epsilon/D} \wedge \nonumber \\
    x_4 & = & x_6 . x_5 \wedge \nonumber \\
        &   & \bigwedge_{i=1}^6 x_i \in L_i, 
\end{eqnarray}
where $A$, $B$, $C$, $D$ are the subsets of $\Sigma(n)$ defined above, and
correspondingly for the regular expressions $L_1,\dots,L_6$.  
Clearly, $\varphi$ is of the form required by the statement of the
proposition. 

Assume first that the constraint $\varphi$ 
is satisfiable via some mapping $\sigma : \{x_1,\dots,x_6\} \to
        \Sigma(n)^*$.  Then $\sigma(x_1)$ is a word in the language
        $(\bigcup_{1 \leq i \leq n} a_i \cdot i)^+$ and $\sigma(x_4)$
        is a word in the language $(\bigcup_{1 \leq j \leq n} b_j \cdot
        j)^+$. Furthermore, $\sigma(x_4) = \sigma(x_5)$ due to the
        fact that $\sigma(x_4) = \sigma(x_6) . \sigma(x_5)$ and
        $\sigma(x_6)$ is the empty word (since $\sigma(x_6)$ belongs
        to $L_6 = \epsilon$). 

Assume that $\sigma(x_1) = a_{i_1} {i_1} \cdots a_{i_k} i_k$, where $1
\leq i_\l \leq n$ for each $1 \leq \l \leq k$, 
and $\sigma(x_4) = b_{j_1} {j_1} \cdots b_{j_p} j_p$, where $1
\leq j_\l \leq n$ for each $1 \leq \l \leq p$.  
Since \begin{eqnarray*} 
\sigma(x_2) & = & \replace{\sigma(x_1)}{\epsilon/A} \\ 
& = & \replace{\sigma(x_4)}{\epsilon/C} \\ 
& = & \replace{\sigma(x_4)}{\epsilon/A},
\end{eqnarray*} it is the case that 
$a_{i_1} \cdots a_{i_k} = b_{j_1} \cdots b_{j_p}$. Similarly, 
since \begin{eqnarray*} 
\sigma(x_3) & = & \replace{\sigma(x_1)}{\epsilon/B} \\ & = &
\replace{\sigma(x_5)}{\epsilon/D} \\ 
& = & \replace{\sigma(x_4)}{\epsilon/D} \\ 
& = & \replace{\sigma(x_4)}{\epsilon/B},\end{eqnarray*} 
we have that $i_1 \cdots i_k =
j_1 \cdots j_p$. We conclude that $i_1,\dots,i_k$ is a solution for
the PCP instance defined by $a_1,\dots,a_n$ and
$b_1,\dots,b_n$.  

Assume on the other hand that $i_1,\dots,i_k$ is a solution for
this PCP instance. Then it is easy to check that the constraint 
$\varphi$ defined above 
is satisfiable via the mapping $\sigma : \{x_1,\dots,x_6\} \to
        \Sigma(n)^*$ which is defined as follows: (1)  
 $\sigma(x_1) = a_{i_1} i_1 \cdots a_{i_k} i_k$, (2) 
$\sigma(x_2) =  a_{i_1} \cdots a_{i_k}
=  b_{i_1} \cdots b_{i_k}$, (3) 
$\sigma(x_3) = i_1 \cdots i_k$, (4) $\sigma(x_4) = \sigma(x_5) = 
b_{i_1} i_1 \cdots b_{i_k} i_k$, and (5) 
$\sigma(x_6) = \epsilon$. 

\subsection*{Another Mutation XSS example: }
\begin{example}
    \em
    This example is an adaptation of \cite[Listing 7]{Stock14}. [The original
    example no longer works on modern browsers.] This example has the same
    spirit as Example \ref{ex:mxss1}. Consider the following 
    JavaScript.
\begin{flushleft}
    \footnotesize
    \begin{verbatim}
var title = document.getElementById("node1");
var x = document.getElementById("node");

z = goog.string.escapeString(z); // sanitise
title.innerHTML = z; 
var name="blah";
var code = '<iframe id="' + title.innerHTML + '"';
code += ' name="' + name +'"';
code += 'src="http://www.w3schools.com"></iframe>';
x.innerHTML = code;
    \end{verbatim}
\end{flushleft}
The HTML file contains the following two lines:
\begin{flushleft}
    \footnotesize
    \begin{verbatim}
<div id="node1" style="font-size: h1">Title</div>
<div id="node">iFrame to add</div>
    \end{verbatim}
\end{flushleft}
The purpose of the JavaScript is to add a user-specified title (specified
in the variable \texttt{z}) to the page, and modifies the id attribute of the 
iframe element accordingly. Since the variable \texttt{z} is
untrusted, it is first sanitised before use. Consider the value
\verb+&#34; onload=&#34;alert(1)+ of the variable \texttt{z}. In this case,
it would remain unmodified by \texttt{escapeString} and is decoded into
\verb+" onload="alert(1)+ by innerHTML mutation. The browser would execute
the code \verb+alert(1)+, which was unintended by the developer.

To analyse such a vulnerability within $\AUWR$, we need to analyse whether 
\texttt{x.innerHTML} could take any string value of the form:
{\footnotesize 
\begin{verbatim}
e2 = /<iframe id="("|[^"]*[^"\\]")
           [a-zA-Z][a-zA-Z0-9]*="("|[^"]*[^"\\]")
           name="("|[^"]*[^"\\]")
           src="http:\/\/www.w3schools.com"><\/iframe>/
\end{verbatim}}
\noindent
The regular expression essentially implies that the constructed string
has an extra attribute.
[The reader could use \url{https://regex101.com/} to experiment.]
Notice that the above script is not a stright-line program. To this end,
we first convert it into the static-single assignment form in the standard
way:
\begin{flushleft}
    \footnotesize
    \begin{verbatim}
var title = document.getElementById("node1");
var x = document.getElementById("node");

newz = goog.string.escapeString(z); // sanitise
title.innerHTML = newz; 
var name="blah";
var code = '<iframe id="' + title.innerHTML + '"';
code1 = code + ' name="' + name +'"';
code2 = code1+'src="http://www.w3schools.com"></iframe>';
x.innerHTML = code2;
    \end{verbatim}
\end{flushleft}
Analysis against the attack pattern \texttt{e2} can then be expressed
in $\AUWR$ as a conjunction of the following constraints:
\begin{itemize}
    \item $\text{newz} = R_2(z)$
    \item $\text{\texttt{title.innerHTML}} = R_3(\text{newz})$
    \item $\text{name}=\text{newz}$
    \item $\text{code} = w_1\cdot \text{\texttt{title.innerHTML}} \cdot w_2$
    \item $\text{code1} = \text{code} \cdot w_3 \cdot \text{name} \cdot w_4$
    \item $\text{code2} = \text{code1} \cdot w_4$
    \item $\text{\texttt{x.innerHTML}} = R_3(\text{code2})$
    \item \texttt{x.innerHTML} matches \texttt{e2}
\end{itemize}
Here $R_2$ is the transducer implementing \texttt{escapeString},
while $R_3$ is a transducer
implementing the implicit browser transductions upon \texttt{innerHTML}
assignments. Observe that the resulting constraint is in $\AUWR$.
\qed
\end{example}

\subsection*{Proof of Lemma \ref{lm:correctness-hypothesis}:} 

The proof of Lemma \ref{lm:correctness-hypothesis} is by 
induction on the position of $y$ in the topological sort $\prec$ of $\G(\varphi)$.

\smallskip 

\underline{Base cases}: A source node $y$ in $\G(\varphi)$. In which case,
by construction the formula $\varphi_y'$ is exactly the same formula as $\varphi_y$ 
except that the variable $y$ is replaced by $y(1)$. (In particular, 
$\varphi_y = \{P(y)\}$, for $P$ a regular language, and therefore $\varphi'_y =
\{P(y(1))\}$). 
So, the claim holds.

\smallskip 


\underline{First inductive case}: \, 
A non-source node $y$, where $y = y_1\ldots y_n$ is the relational constraint 
witnessing
the incoming edge to $y$ and $P(y)$ is the regular constraint for $y$. In this
case, let $z$ be the immediate predecessor of $y$ in the topological sort $\prec$.
Among others, observe that $y_1,\ldots,y_n \preceq z$. To prove the desired claim
for $y$, we will now apply the induction hypothesis on $z$. 

In one direction, if $\iota : Z_y \to \Sigma^*$ 
is a satisfying assignment for $\varphi_y$, then the 
restriction 
$\iota_{|_{Z_z}}$ of $\iota$ to $Z_z$ is a satisfying
assignment for $\varphi_z$ (since $\varphi_y = \varphi_z \wedge y = y_1\ldots y_n
\wedge P(y)$). By induction then, there exists a satisfying assignment $\iota': 
Z_z' \to
\Sigma^*$ for $\varphi_z'$ such that
\[
    \iota_{|_{Z_z}}(u) \ = \ \iota'(u(1))\circ \cdots \circ 
                                \iota'(u({\max}_{u}))
\]
for each $u \in Z_z$. Therefore:  
\begin{multline*}\iota(y) \ = \ \iota(y_1)\ldots \iota(y_n) \ = \\
\iota'(y_1(1))\ldots \iota'(y_1({\max}_{y_1}))\ldots \iota'(y_n(1)\ldots
\iota'(y_n({\max}_{y_n})).
\end{multline*} 
For each $k \in \{1,\dots,{\max}_y\}$ and integers $i,l$
satisfying $\nu(k) = (i,l)$, the formula $\varphi_y'$ contains the constraint
$y(k) = y_i(l)$, and so we will simply extend the assignment $\iota': Z_z' \to
\Sigma^*$ to $Z'_y$ by setting $\iota'(y(k)) :=
\iota'(y_i(l))$. This shows that $\iota(u) = \iota'(u(1))\circ \cdots \circ 
                                \iota'(u({\max}_{u}))$ 
for each $u \in Z_y$. It also shows that 
$\iota' : Z'_y \to \Sigma^*$
satisfies $\NewStrCon{y}$. In addition, if $\pi$ is an accepting run
of $P$ on $\iota(y)$, then we may split $\pi$ into $\max_y$ segments 
$\pi_1,\ldots,\pi_{\max_y}$, where the initial (resp. final) state of each $\pi_i$
is a state $q_{i-1}$ (resp. $q_i$), such that the run $\pi_i$ is also an
accepting run of $P_{[q_{i-1},q_i]}$ on $\iota'(y(i))$. This shows
that 
$\iota' : Z'_y \to \Sigma^*$ satisfies 
$\bigwedge_{i=1}^{\max_y} P_{[q_{i-1},q_i]}(y(i))$ for some
$\max_y$-splitting $q_0,\dots,q_{\max_y}$ of $P$, and hence that it
satisfies $\NewRegCon{y}$. We conlude
then that $\iota' : Z'_y \to \Sigma^*$ satisfies $\varphi'_y$.  

The converse case can be proven by noticing that the chain of reasoning
in the previous paragraph can be reversed.

\smallskip 

\underline{Second inductive case}:
\, A non-source node $y$, where $R(x,y)$ is the string constraint witnessing
the incoming edge to $y$ and $P(y)$ is the regular constraint for
$y$. Again, let $z$ be the immediate predecessor of $y$ in the topological sort $\prec$.
Thus, $x \preceq z$. To prove the desired claim
for $y$, we apply the induction hypothesis on $z$. 



In one direction, assume $\iota : Z_y \to \Sigma^*$ 
is a satisfying assignment for $\varphi_y$. As in the previous case,
we can conclude by induction hypothesis that there exists a satisfying assignment $\iota': 
Z_z' \to
\Sigma^*$ for $\varphi_z'$ such that
\[
    \iota_{|_{Z_z}}(u) \ = \ \iota'(u(1))\circ \cdots \circ 
                                \iota'(u({\max}_{u}))
\]
for each $u \in Z_z$ (and, in particular, for $u = x$).
Since $\iota$ is a satisfying assignment for $\varphi_y$, it also follows that
there exists a run $\pi$ of $R$ witnessing $(\iota(x),\iota(y)) \in R$. We may
split this run $\pi$ into $\max_x$ segments $\pi_1,\ldots,\pi_{\max_x}$,
where the initial (resp. final) state of $\pi_i$ is some state $p_{i-1}$
(resp. $p_i$), such that for some $w_1,\ldots,w_n \in \Sigma^*$ such that
$\iota(y) = w_1\ldots w_n$, the run $\pi_i$ witnesses that 
$(\iota'(x_i),w_i) \in R_{[p_{i-1},p_i]}$. (Note that there might be more than one way
to split $\pi$ and $\iota(y)$ while satisfying the aforementioned condition,
e.g., if $\pi_1 = p_0 \tran{a/bb} p_1 \tran{\epsilon/a} p_2$ and $\pi_2 =
p_2 \tran{b/ab} p_3$, then the splitting $\pi_1' = p_0 \tran{a/bb} p_1$ and
$\pi_2' = p_1 \tran{\epsilon/a} p_2 \tran{b/ab} p_3$ will also satisfy the
aforementioned condition). We will simply extend the assignment $\iota'$ to
$Z'_y$ by setting $\iota'(y(k)) = w_k$ for each $k \in
\{1,\dots,\max_x\}$. Clearly then, $\iota(u) = \iota'(u(1))\circ \cdots \circ 
                                \iota'(u({\max}_{u}))$ 
for each $u \in Z_y$. Also, by definition $\iota'$ satisfies 
$\bigwedge_{i=1}^{{\max}_y} y(i) = R_{[p_{i-1},p_i]}(x(i))$ for some
$\max_y$-splitting $p_0,\dots,p_{\max_y}$ of $R$, and hence it
satisfies $\NewStrCon{y}$. 
Now, let $\pi$ be an accepting run of of $P$ on $\iota(y) = w_1 \ldots
w_n$. As in the first inductive case, we take the segment 
$\pi_i$ of $\pi$ that operates on $w_i$. If $p_{i-1}$ (resp. $p_i$) is the 
initial (resp. final) state in the run segment $\pi_i$, then $\iota'$ satisfies
$\bigwedge_{i=1}^{\max_y} P_{[p_{i-1},p_i]}(y(i))$. This shows that
$\iota'$ satisfies $\NewStrCon{y}$. We conclude that $\iota'$
satisfies $\varphi'_y$. 

As in the first inductive case, the converse case can be proven by noticing that 
the chain of reasoning in the previous paragraph can be reversed.

\subsection*{Proof of the lower bound of Theorem \ref{th:expspace}: }

We prove that checking satisfiability of string 
constraints in $\AUWR$ is
$\EXPSPACE$-hard. We reduce from the acceptance problem for a
deterministic Turing machine $\M$ that works in space $2^{cn}$, for $c
> 1$. That is, we provide a polynomial time reduction that, given
an input $w$ to $\M$, it constructs a constraint $\varphi(w)$ in
$\AUWR$ such
that $w$ is accepted by $\M$ if and only if $\varphi(w)$ is
satisfiable. 

Let us assume that $\M = (\Sigma,Q,q_0,q_f,\delta)$, where (i) 
$\Sigma$ is a
finite alphabet which contains the blank symbol $\flat$, 
(ii) $Q$ is the finite set of states, (iii) $q_0$ is the
initial state, (iv) $F$ is the set of final states, and (v) 
$\delta :
Q \times \Sigma \to Q \times \Sigma \times \{-1,0,1\}$ is the
transition function. We assume without loss of generality that
(a) $\Sigma \cap \{0,1\} = \emptyset$, (b) $F =
\{q_f\}$ is a singleton, and (c) before accepting the machine erases
its tape leaving the head in the first position. 

We define $\Sigma(Q) := \Sigma \cup (\Sigma \times Q) \cup \{0,1\}$. 
We use this
alphabet to represent configurations of $\M$. Consider, for instance, that $\M$
is in a configuration $c$ in which the $i$-th cell of the tape contains
symbol $a_i \in \Sigma$, for each $0 \leq i \leq 2^{cn}-1$, the
head is scanning cell $0 \leq j \leq 2^{cn}-1$, and the machine is in
state $q \in Q$. Then $c$ is represented as the following word over
$\Sigma(Q)$: 
$$[0] \, a_0 \, [1] \, a_1 \, [2] \, \cdots \, [j] \, (a_j,q) \, [j+1]
\, \dots \, [2^{cn}-1]
a_{2^{cn}-1},$$
where for each $0 \leq i \leq 2^{cn} - 1$ we have that $[i]$ is the
$cn$-bit representation of the integer $i$ over alphabet $\{0,1\}$. 


Let $w$ be an input to $\M$ (i.e., a word over
$\Sigma$) of length $n \geq 0$.  We now explain how to construct
$\varphi(w)$.
To start with, $\varphi(w)$ contains a regular constraint $\bigwedge_{1
\leq i \leq m} L_i(x)$ such that $L_1(x) \wedge \dots \wedge L_m(x)$ is satisfiable by a word
$x$ if and only the following holds: 
\begin{enumerate} 
\item 
The word $x$ satisfies the regular expression: 
$$\big([0] \, \Sigma(Q) \, [1] \, \Sigma(Q) \, [2] \, \cdots \, [2^{cn}-1]
\, \Sigma(Q) \, \$ \big)^*.$$
In other words, $x$ represents the concatenation of several $cn$-bit
counters separated by the delimiter \$. Each address $[i]$ in one of
these bit couters is followed by a symbol from $\Sigma(Q)$.  

\item The first counter in $x$ encodes the initial
configuration of $\M$ with input $w$. That is, if $w = a_0a_1 \cdots
a_{n-1}$ then the following word is a prefix of $x$:   
$$[0] \, (a_0,q_0) \, [1] \, a_1 \, \cdots \, [n-1] \,
a_{n-1} \, [n] \, \flat \, \cdots \, 
[2^{cn}-1] \, \flat \, \$ $$
Recall that $q_0$ is the initial state of $\M$ and $\flat$ is the
blank symbol. 

\item The last counter of $x$ encodes a final 
configuration of $\M$. That is, the following word is a suffix of $x$: 
$$[0] \, (\flat,q_f) \, [1] \, \flat \, [2]  \, \flat \, \cdots \, 
[2^{cn}-1] \, \flat.$$

\end{enumerate} 
Using standard ideas on how to enconde $n$-bit counters (see, e.g., \cite{Bor,Kozen77}),
it is possible to show that a set $\{L_1,\dots,L_m\}$ of NFAs
that satisfy the conditions specified above can be constructed in
polynomial time from $\M$ and $w$. 

The intuitive idea of the reduction is to codify in the word $x$ an
accepting run of $\M$ on input $w$.  
Let us assume that $x$ is of the form $c_0 \$ c_! \$ \cdots \$
c_p \$$, for $p \geq 0$, 
where for each $0 \leq i \leq p$ we have that $c_i$ encodes a
$cn$-bit counter. We know that $c_0$ and $c_p$
represent an initial and final configuration, respectively, of $\M$ on
input $w$. Thus, in order to check that $x$ in fact represents an
accepting run of $\M$ on $w$ we only need to verify that for each $0
\leq j \leq p-1$ it is the case that $c_{j+1}$ encodes the
configuration which is obtained
from $c_j$ by one application of the transition function $\delta$. 

As usual, we check this locally for each {\em block} of three
consecutive addresses in the $c_j$'s. Formally, for each $0 \leq i
\leq 2^{cn} -1$ and $0 \leq j \leq p$, we define $c_j[i]$ to be the
unique substring of $c_j$ of the form $[i]a$, for $a \in \Sigma(Q)$.  
Clearly
then, checking whether the word $x = c_0 \$ c_1 \$ \cdots \$ c_p \$$
codifies an accepting run of $\M$ on input $w$ is equivalent to
checking that for each $1 \leq i \leq 2^{cn}-2$ and $0 \leq j < p$ the
block $c_{j+1}[i-1] c_{j+1}[i] c_{j+1}[i+1]$ follows from $c_j[i-1]
c_j[i] c_j[i+1]$ according to the transition function $\delta$. For
instance, if $$c_j[i-1] c_j[i] c_j[i+1] \ = \ [i-1] \, a \, [i] \,
(b,q) \, [i+1] \, c,$$ 
where
$a,b,c \in \Sigma$ and $q \in Q$, and $\delta(q,b) = (q',b',-1)$, then
$$c_{j+1}[i-1] c_{j+1}[i] c_{j+1}[i+1] \ = \ [i-1] \, (a,q') \, [i] \,
b' \, [i+1] \, c.$$
This is precisely the role of the remaining constraints in
$\varphi(w)$, which we define next.

First, we add to the constraint $\varphi(w)$ the string constraint 
$$y \ = \ x \# x
\# x \# x \# x \# x$$ that takes six copies of
the word $x$, separates them with a new delimiter $\#$, and assigns
the result to the variable $y$. 
We then apply 
a transducer $R$ to $y$ to define a new variable $z$. Let us recall that
$x$ is of the form $c_0 \$ c_1 \$ \dots \$ c_p \$ $, where each
$c_i$ encodes a $cn$-bit counter. The
transducer $R$ then converts $y = x \# x
\# x \# x \# x \# x$ into $z = x_1 \# x_2
\# x_3 \# x_4 \# x_5 \# x_6$, where:   

\begin{enumerate}
\item
 $x_1$ is exactly as
$x$ except that now we have replaced each counter $c_j$ with a
new counter $c^1_j$ that only preserves from $c_j$ information about
addresses $[0],[1],[2],[6],[7],[8],\dots$
\item Same for $x_2$ and $c_2^j$, but now $c_2^j$
preserves information about
addresses $[1],[2],[3],[7],[8],[9],\dots$
\item Same for $x_3$ and $c_3^j$, but now $c_3^j$
preserves information about
addresses  $[2],[3],[4],[8],[9],[10],\dots$ 
\item Same for $x_4$ and $c_4^j$, but now $c_4^j$
preserves information about
addresses $[3],[4],[5],[9],[10],[11],\dots$
\item Same for $x_5$ and $c_5^j$, but now $c_5^j$
preserves information about
addresses  $[4],[5],[6],[10],[11],[12],\dots$
\item Same for $x_6$ and $c_6^j$, but now $c_6^j$
preserves information about
addresses  $[5],[6],[7],[11],[12],[13],\dots$
\end{enumerate} 
We also assume that transducer $R$ 
deletes from each counter $c_j^l$ in $u_l$ ($0 \leq j \leq
p$ and $1 \leq l \leq 6$) the ``incomplete'' blocks before a delimiter \$, i.e., the
suffixes that do not form a complete block of
three consecutive addresses. (For instance, if $2^{cn} = 8$ then each
$c_j^2$ is of the form $[1] a [2] b [3] c [7] d$. In this case, $R$
simply deletes the suffix $[7] d$ completely). 

Notice that each block of three consecutive addresses appears in one,
and only one, of the words
$x_l$, for $1 \leq l \leq 6$. In particular, the information from $x$
about addresses $[i-1],[i],[i+1]$, for $1 \leq i \leq 2^{cn}-2$, is
preserved in $x_j$, where $j$ is the remainder obtained by dividing
$i$ by 6 (assuming that if $i$ is divisible by 6 then this
remainder is 6). Further, it is not hard to see
how the transducer $R$ can be constructed in polynomial time from our
input.   

We now extend the formula $\varphi(w)$ by including the conjunction of
the following atomic string constraints:  
\begin{align*}
u_{i} \ & = \ z_{i-1} \#_i z_{i-1} \ \ \ & (1 \leq i \leq cn-1) \\
z_{i} \ & = \ S_i(u_{i-1}) & (1 \leq i \leq cn),  
\end{align*} 
where $z_0 = z$, the symbol $\#_i$ is a fresh
delimiter, and $S_i$ is a transducer we define
below. That is, $u_i$ consists of two copies of $z_{i-1}$ separated by
the delimiter $\#_i$ and $z_i$ is obtained from $u_{i-1}$ by applying
the transducer $S_i$. We stop this process after $cn$ steps. 

For the sake of readability, we will keep the explanation at the
intuitive level.  We start by explaining how the transducer $S_1$ is
defined. By definition, $u_1 = z_0 \#_1 z_0 = z \#_1 z$ and $z = x_1
\# x_2 \# x_3 \# x_4 \# x_5 \# x_6$. Then $S_1$ takes the the copy of
$x_l = c_0^l \$ \dots \$ c_p^l$, for $1 \leq l \leq 6$, that appears
before the delimiter $\#_1$ in $u_1$, and turns it into a new word
$x'_l = d_0^l \$ \dots \$ d_p^l$, where each $d_j^l$ is obtained by
performing the following modifications over $c_j^l$: If $c_j^l$
consists only of a block of three consecutive addresses, then leave as
it is (i.e., $d_j^l = c_j^l$); otherwise, remove every other block of
three consecutive addresses from $c_j^l$. For the copy of $x_l$ that
occurs after $\#_1$, the transducer $S_1$ transforms
it into a new word $x''_l$ by doing the opposite: If $c_j^l$ contains
more than one block of three consecutive addresses, then it removes
the first block, keeps the second one, removes the third one, etc. 
It is not hard to see that $S_1$ can be constructed in polynomial
time from our input.

Notice then that 
$$z_1 = x'_1 \# x'_2 \# x'_3 \# x'_4 \# x'_5 \# x'_6 \#_1 
x''_1 \# x''_2 \# x''_3 \# x''_4 \# x''_5 \# x''_6,$$ where, 
for instance, $x'_1$ corresponds to the restriction
of $x$ to addresses $[0],[1],[2],[12],[13],[14],\dots$, while $x''_1$
corresponds to the restriction of $x$ to 
$[6],[7],[8],[18],[19],[20],\dots$. In the same way, $x'_2$ 
corresponds to the restriction of $x$ to addresses
$[1],[2],[3],[13],[14],[15],\dots$, while $x''_2$ 
corresponds to the restriction of $x$ to addresses
$[7],[8],[9],[19],[20],[21],\dots$. In general, the restriction of $x$
to addresses $[i-1],[i],[i+1]$, for $1 \leq i \leq 2^{cn}-2$, is
contained in (i) $x'_j$ if and only if the remainder of $i$ divided by 12 is $1 \leq
j \leq 6$, and in (ii) $x'_j$ if and only if the remainder of $i$ divided by 12 is $j +
6$, for $1 \leq j \leq 6$. 



In general, we will assume inductively that $z_l$, for $1 \leq l \leq
cn - 1$, is constructed from words $w_1,w_2,\dots,w_{6 \cdot 2^i}$
over $\Sigma(Q) \cup \{\$\}$ -- suitably separated by delimiters
$\#,\#_1,\dots,\#_{i}$ according to the structure provided by the
constraints that define the $z_l$'s and $u_l$'s -- in such a way that
$w_j$ corresponds to the restriction of $x$ to addresses
$[i-1],[i],[i+1]$ ($1 \leq i \leq 2^{cn}-2$) for which the
remainder of $i$ divided by $6 \cdot 2^i$ is precisely $j$.

Then $u_{l+1} = z_l \#_{l+1} z_l$, and we define $S_{l+1}$ similarly to
$S_1$. That is, $S_{l+1}$ takes the copy of
$w_j =  t_0 \$ \dots \$ t_p$  (where each $t_h$ is a word over
$\Sigma(Q)$) 
that appears
before the delimiter $\#_{l+1}$ in $u_{l+1}$, and turns it into a new word
$w'_j = t'_0 \$ \dots \$ t'_p$, where each $t'_h$ is obtained by
performing the following modifications over $t_h$: If $t_h$
consists only of a block of three consecutive addresses, then leave as
it is (i.e., $t'_h = t_h$); otherwise, remove every other block of
three consecutive addresses from $t_h$. For the copy of $w_j$ that
occurs after $\#_{l+1}$, the transducer $S_{l+1}$ transforms
it into a new word $w''_j$ by doing the opposite: If $t_h$ contains
more than one block of three consecutive addresses, then it removes
the first block, keeps the second one, removes the third one, etc. 
It is not hard to see then that $z_{l+1}$ consists of words
$v_1,v_2,\dots,v_{6 \cdot 2^{i+1}}$ over $\Sigma(Q) \cup \{\$\}$ -- suitably separated by delimiters
$\#,\#_1,\dots,\#_{l+1}$ according to the structure provided by the
constraints that define the $z_l$'s and $u_l$'s -- in such a way that
$v_j$ corresponds to the restriction of $x$ to addresses
$[i-1],[i],[i+1]$ ($1 \leq i \leq 2^{cn}-2$) for which the
remainder of $i$ divided by $6 \cdot 2^{i+1}$ is precisely $j$. This
proves our inductive case. 

In particular then, $z_{cn}$ consists of words
$v_1,v_2,\dots,v_{6 \cdot 2^{cn}}$ over $\Sigma(Q) \cup \{\$\}$ -- suitably separated by delimiters
$\#,\#_1,\dots,\#_{cn}$ according to the structure provided by the
constraints that define the $z_l$'s and $u_l$'s -- in such a way that
$v_j$, for $1 \leq j \leq 2^{cn}-2$, corresponds to the restriction of $x$ to addresses
$[j-1],[j],[j+1]$ (and nothing else). That is, $v_j$ is of the form  
\begin{multline*}
c_0[j-1] c_0[j] c_0[j+1] \$ c_1[j-1] c_1[j] c_1[j+1] \$ \cdots \\
\$ \, c_p[j-1] c_p[j] c_p[j+1],
\end{multline*} Further, the $v_j$'s such that $j >
2^{cn}-2$ also correspond to the restriction of $x$ to addresses
$[i-1],[i],[i+1]$ (and nothing else), for some $1 \leq i \leq
2^{cn}-2$.  

Finally, we only have to check in each subword $v_j$ ($1 \leq j \leq 6
\cdot 2^{cn}$) of $z_{cn}$ of the form 
\begin{multline*}
c_0[i-1] c_0[i] c_0[i+1] \$ c_1[i-1] c_1[i] c_1[i+1] \$ \cdots \\
\$ \, c_p[i-1] c_p[i] c_p[i+1],
\end{multline*} 
for $1 \leq i \leq
2^{cn}-2$, that $c_{l+1}[i-1] c_{l+1}[i] c_{l+1}[i+1]$ is obtained
from $c_{l}[i-1] c_{l}[i] c_{l}[i+1]$ by applying the
transition function $\delta$, for each $0 \leq l < p$. 
It is easy to see how to construct an
NFA $L$ which verifies this for each subword $v_j$ of the form above 
(see, e.g.,
\cite{Kozen77}). We then add to $\varphi(w)$ a regular constraint $\A^*(z_{cn})$,
where $\A^*$ is the NFA that checks that each subword $v_j$ in $z_{cn}$ is
accepted by $\A$. The NFA $\A^*$
simply ``restarts'' $\A$ if $\A$ is in an acceptance state while
reading one of the delimiters $\#_i$ or $\#$; otherwise it rejects.

Clearly, $\varphi(w)$ is in $\AUWR$ and 
can be constructed in polynomial time
from the input. Furthermore, $\M$ accepts $w$ if and only if
$\varphi(w)$ is satisfiable. This concludes the proof.
 
}

\end{document}